%% file: cycle.tex
\documentclass[11pt,letterpaper]{article}
\input{sections/cycle_preamble}

\usepackage[margin=1in]{geometry}

\makeatletter
\def\blfootnote{\gdef\@thefnmark{}\@footnotetext}
\makeatother

\begin{document}
	\title{Approximating Min-Mean-Cycle for low-diameter graphs in near-optimal time and memory}
	\author{Jason M. Altschuler \and Pablo A. Parrilo}
	\date{}
	\maketitle
	\blfootnote{The authors are with the Laboratory for Information and Decision Systems (LIDS), Massachusetts Institute of Technology, Cambridge MA 02139. Work partially supported by NSF AF 1565235, NSF Graduate Research Fellowship 1122374, and a TwoSigma PhD Fellowship.}

\input{sections/abstract}

\input{sections/intro}

\input{sections/prelim}

\input{sections/framework}

\input{sections/bal}

\input{sections/round}

\input{sections/alg_final}

\input{sections/experiments}

\section*{Acknowledgements}
We thank Mina Dalirrooyfard, Jonathan Niles-Weed, and Joel Tropp for helpful conversations.

\appendix
\input{sections/app}

\footnotesize
\addcontentsline{toc}{section}{References}
\bibliographystyle{abbrv}
\bibliography{cycle}

\end{document}

%% file: sections/cycle_preamble.tex
\usepackage{amsmath,amssymb,amsthm}
\usepackage{MnSymbol} 
\usepackage{graphicx,color}
\usepackage[hyphens]{url}
\usepackage{dsfont}
\usepackage{booktabs}
\usepackage[square,sort,numbers]{natbib}
\usepackage[sort,nocompress,noadjust]{cite}
\usepackage[normalem]{ulem}
\usepackage{mathtools}
\usepackage[nameinlink,capitalize]{cleveref}
\usepackage{multirow}
\usepackage{algorithm}
\usepackage[noend]{algpseudocode}
\usepackage{xspace}

\usepackage{caption}
\usepackage{mwe}

\usepackage{tabularx}
\newcolumntype{L}{>{\raggedright\arraybackslash}X}

\usepackage{etoolbox}

\usepackage{subcaption}
\usepackage{tikz}
\numberwithin{equation}{section}
\newtheorem{theorem}{Theorem}[section]

\newtheorem{lemma}[theorem]{Lemma}
\newtheorem{remark}[theorem]{Remark}

\newtheorem{obs}[theorem]{Observation}

\newtheorem{defin}[theorem]{Definition}

\makeatletter
\newenvironment{subtheorem}[1]{%
	\def\subtheoremcounter{#1}%
	\refstepcounter{#1}%
	\protected@edef\theparentnumber{\csname the#1\endcsname}%
	\setcounter{parentnumber}{\value{#1}}%
	\setcounter{#1}{0}%
	\expandafter\def\csname the#1\endcsname{\theparentnumber.\Alph{#1}}%
	\ignorespaces
}{%
	\setcounter{\subtheoremcounter}{\value{parentnumber}}%
	\ignorespacesafterend
}
\makeatother
\newcounter{parentnumber}

\newcommand{\cF}{\mathcal{F}}

\newcommand{\cM}{\mathcal{M}}

\newcommand{\cT}{\mathcal{T}}

\newcommand{\R}{\mathbb{R}}

\newcommand{\Rn}{\R^n}

\newcommand{\Rnn}{\R^{n \times n}}
\newcommand{\Rp}{\R_{\geq 0}}

\newcommand{\Rpnn}{\R_{\geq 0}^{n \times n}}

\newcommand{\Prob}{\mathbb{P}}
\newcommand*{\E}{\mathbb{E}}
\newcommand{\supp}{\operatorname{supp}}

\providecommand{\diag}{\operatorname{\mathbb{D}}}

\newcommand{\bone}{\mathbf{1}}
\newcommand{\zero}{\mathbf{0}}

\newcommand*{\Otilde}{\tilde{O}}

\DeclareMathOperator{\logalp}{\log\tfrac{1}{\alpha}}

\newcommand*{\poly}{\mathrm{poly}}
\newcommand*{\polylog}{\mathrm{polylog}}

\newcommand*{\eps}{\varepsilon}

\newcommand{\plusminus}{\raisebox{.2ex}{$\scriptstyle\pm$}}

\renewcommand{\leq}{\leqslant}
\renewcommand{\geq}{\geqslant}

\newcommand{\MMC}{\textsf{MMC}}

\newcommand{\SSSP}{\textsf{SSSP}}

\newcommand{\mAM}{\mu}
\newcommand{\mAMG}{\mu(G)}

\newcommand{\DE}{\Delta_E}
\newcommand{\FE}{\cF_E}

\DeclareMathOperator*{\smin}{smin}
\DeclareMathOperator*{\smineta}{{\textstyle{smin_\eta}}}

\newcommand{\diamG}{d}
\newcommand{\dt}{\tilde{d}}

\newcommand{\Algbal}{\texttt{AMMC}}

\newcommand{\ApproxBalance}{\texttt{ABAL}}

\newcommand{\Approxdiam}{\texttt{ADIAM}}

\newcommand{\Roundtocirc}{\texttt{RoundCirc}}
\newcommand{\Quantroundtocirc}{\texttt{RoundQCirc}}
\newcommand{\Extractcycle}{\texttt{RoundCycle}}

\newcommand{\Tbal}{\cT_{\ApproxBalance}}
\newcommand{\Mbal}{\cM_{\ApproxBalance}}

\newcommand{\wbar}{\bar{w}}

\newcommand{\condK}{\kappa}

\newcommand{\Wmax}{w_{\max}}
\newcommand{\Wmin}{w_{\min}}

\newcommand{\Feta}{F^{\eta}}

\providecommand{\imb}[1]{\delta(#1)}
\providecommand{\imbind}[2]{\delta_{#1}(#2)}
\newcommand{\imbP}{\imb{P}}
\newcommand{\imbQ}{\imb{Q}}
\newcommand{\imbPtilde}{\imb{\tilde{P}}}

\newcommand{\Ftilde}{\tilde{F}}

\newcommand{\Rmax}{\mathbb{R}_{\max}}

%% file: sections/abstract.tex
\begin{abstract}
	\small
	We revisit Min-Mean-Cycle, the classical problem of finding a cycle in a weighted directed graph with minimum mean weight. Despite an extensive algorithmic literature, previous work falls short of a near-linear runtime in the number of edges $m$. We propose an approximation algorithm that, for graphs with polylogarithmic diameter, achieves a near-linear runtime. In particular, this is the first algorithm whose runtime scales in the number of vertices $n$ as $\tilde{O}(n^2)$ for the complete graph. Moreover---unconditionally on the diameter---the algorithm uses only $O(n)$ memory beyond reading the input, making it ``memory-optimal''. Our approach is based on solving a linear programming relaxation using entropic regularization, which reduces the problem to Matrix Balancing---\'a la the popular reduction of Optimal Transport to Matrix Scaling. 
	The algorithm is practical and simple to implement.
\end{abstract}

%% file: sections/intro.tex
\section{Introduction}\label{sec:intro}

Let $G = (V,E,w)$ be a weighted directed graph (digraph) with vertices $V$, directed edges $E \subseteq V \times V$, and edge weights $w : E \to \R$. The \textit{mean weight} of a cycle $\sigma$ is the arithmetic mean of the weights of the cycle's constituent edges, denoted $\wbar(\sigma) := \tfrac{1}{|\sigma|} \sum_{e \in \sigma} w(e)$. The \emph{Min-Mean-Cycle} problem ($\MMC$ for short) is to find a cycle of minimum mean weight. The corresponding value is denoted 
\begin{align}
\mAMG := \min_{\text{cycle }\sigma\text{ in G}} \wbar(\sigma).
\tag{MMC}
\label{MMC}
\end{align}

\par Over the past half century, $\MMC$ has received significant attention due to its numerous fundamental applications in periodic optimization, algorithm design, and max-plus algebra. Applications in periodic optimization include deterministic Markov Decision Processes and mean-payoff games~\citep{ZwiPat96}, financial arbitrage~\citep{CLRS}, cyclic scheduling problems~\citep{KarOrl81}, and performance analysis of digital systems~\citep{DasIraGup99}, among many others.
 In algorithm design, $\MMC$ provides a tractable option for the bottleneck step in the network simplex algorithm. This has led to the use of $\MMC$ in algorithms for several graph theory problems~\citep{AhuOrl01,OuoMah00}---including, notably, a strongly polynomial algorithm for the Minimum Cost Circulation problem, which includes Maximum Flow as a special case~\citep{GolTar89}. 
In max-plus algebra, which commonly arises in operations research and control theory problems, $\MMC$ characterizes the fundamental spectral theoretic quantities~\citep{BapStaDri93,Gun94}. More recently, $\MMC$ has also arisen in control theory since it captures the growth rate of switched linear dynamical systems with rank-one updates~\citep{AhmPar12,AltPar20lyap}.

\par These myriad applications have motivated a long line of algorithmic work with the goal of solving $\MMC$ efficiently. Remarkably, $\MMC$ is solvable in polynomial time, despite the fact that many seemingly similar optimization problems over cycles are not. Indeed, in sharp contrast, the problem of finding the cycle with minimum \textit{total} weight $\sum_{e \in \sigma} w(e)$ is NP-complete since it can encode the Hamiltonian Cycle problem~\citep[\S8.6b]{Sch03}.

\par Algorithmic advancements over the past half century have led to many efficient algorithms for $\MMC$; details in the prior work section \S\ref{ssec:intro:prev} below. However, previous work falls short of a near-linear runtime in the input sparsity $m := |E|$.
For instance, even in the ``simple'' case where the edge weights are in $\{-1,0,1\}$, the best known runtimes are $O(m \sqrt{n} \log n)$ from~\citep{OrlAhu92}, 
$m^{11/8+o(1)}$ implicit from~\citep{AxiMadVla20},
and $O(n^{\omega} \polylog n)$ implicit from~\citep{San05,YusZwi05}, where $n := |V|$ is the number of vertices, and $\omega \approx 2.37$ is the current matrix multiplication exponent~\citep{Wil14}. These runtimes are incomparable in the sense that which is fastest depends on the graph sparsity (i.e., the ratio of $m$ to $n$). Nevertheless, in all parameter settings, these runtimes are far from linear in $m$. An important algorithmic barrier is that \emph{any} faster runtime---let alone a linear runtime---for solving a natural LP relaxation of $\MMC$ would constitute a major breakthrough in algorithmic graph theory, as it would imply faster algorithms for many well-studied problems (e.g., Shortest Paths with negative weights~\citep[\S8.2]{Sch03}).

\par A primary motivation of this paper is the observation that this complexity barrier is only for \emph{exactly} computing (this LP relaxation of) $\MMC$. Indeed, our main result is that for graphs with polylogarithmic diameter, $\MMC$ can be \textit{approximated} in near-linear\footnote{Throughout, we say a runtime is near-linear if it is $O(m)$, up to polylogarithmic factors in $n$ and polynomial factors in the inverse accuracy $\eps^{-1}$ and the maximum modulus edge weight $\Wmax$. 
} time.

\subsection{Contributions}\label{ssec:intro:contributions}

Henceforth, $G$ is assumed strongly connected; this is without loss of generality for $\MMC$ after a trivial $O(m)$ pre-processing step; see \S\ref{sec:prelim}. We denote the \textit{unweighted} diameter of $G$ by $\diamG$. The notation $\Otilde(\cdot)$ suppresses polylogarithmic factors in the number of vertices $n$, the inverse accuracy $\eps^{-1}$, and the maximum modulus edge weight $\Wmax$.
\par We give the first approximation algorithm for $\MMC$ that, for graphs with polylogarithmic diameter, has near-linear runtime in the input sparsity $m$. In particular, this is the first near-linear time algorithm for the important special cases of complete graphs, expander graphs, and random graphs. (Note also that if the diameter is larger than polylogarithmic, this runtime can still be much faster than the state-of-the-art, depending on the parameter regime.) Moreover, unconditionally on the diameter, this new algorithm requires only $O(n)$ additional memory beyond reading the input\footnote{Storing the input graph takes $\Theta(m)$ memory. To design an algorithm with $o(m)$ memory, we assume $G$ is input implicitly through two oracles: one for finding an adjacent edge of a vertex, and one for querying the weight of an edge; details in \S\ref{sssec:bal:final:mem}.}, which means it is so-called ``memory-optimal'' in the sense that its memory usage is of the same order as the (maximum possible) output size.

\begin{theorem}[Informal version of Theorem~\ref{thm:bal:orand}]\label{thm:bal:orand-intro}
	There is a randomized algorithm ($\Algbal$ on page~\pageref{alg:bal}) that given a weighted digraph $G = (V,E,w)$ and an accuracy $\eps > 0$, finds a cycle $\sigma$ in $G$ satisfying $\wbar(\sigma) \leq \mAMG + \eps$ using $O(n)$ memory beyond reading the input and
	$O( m \diamG^2 (\tfrac{\Wmax}{\eps})^2 \log n )$
	time, both in expectation and with exponentially high probability.
\end{theorem}

This algorithm $\Algbal$ is based on approximately solving an entropically regularized version of an LP relaxation of $\MMC$, followed by rounding the obtained fractional LP solution using a fast, approximate version of the classical Cycle-Cancelling algorithm; details in the overview section \S\ref{ssec:intro:tech}. The entropic regularization approach has two key benefits. First, it effectively reduces the optimization problem to Matrix Balancing---a well-studied problem in scientific computing for which near-linear time algorithms were recently developed~\citep{ZhuLiOliWig17,AltPar20bal,CohMadTsiVla17,OstRabYou16}. At a high-level, this parallels the popular entropic-regularization reduction of Optimal Transport to Matrix Scaling~\citep{Cut13,Wil69}. Second, it enables a compact $O(n)$-size implicit representation of the (na\"ively $O(m)$-size) fractional solution to the LP relaxation. 
\paragraph*{Discussion}

\begin{description}
	\item \underline{Practicality.} $\Algbal$ is practical and simple to implement. This is in contrast to the aforementioned state-of-the-art theoretical algorithms, which rely on (currently) impractical subroutines such as Fast Matrix Multiplication or fast Laplacian solvers, and/or have large constants in their runtimes which can be prohibitive in practice. Indeed, there is currently a large discrepancy between the state-of-the-art $\MMC$ algorithms in theory and in practice: the algorithms with best empirical performance have worst-case runtimes no better than $\Omega(mn)$; see the experimental surveys~\citep{Cha01,Das04,DasIraGup99,GGTW09}. In Section~\ref{sec:experiments}, we provide preliminary numerical simulations demonstrating that in practice, $\Algbal$ can compute high-quality solutions in essentially $O(m)$ linear runtime and for larger problem sizes than the state-of-the-art algorithms implemented in the popular, heavily-optimized C++ software package LEMON~\citep{lemon}.
	\item \underline{Multiplicative approximation.} If all edge weights are positive, then the \emph{additive} approximation of $\Algbal$ also yields a \emph{multiplicative} approximation. (If the edge weights are not all positive, then it is impossible to compute any multiplicative approximation in near-linear time, barring a major breakthrough in algorithmic graph theory, namely faster algorithms for the classical Negative Cycle Detection problem~\citep[\S1.2]{CHKLR14}.) Specifically, if all edge weights lie in $[\Wmin,\Wmax]$ for $\Wmin > 0$, then we can find a cycle $\sigma$ satisfying $\wbar(\sigma) \leq (1+\eps)\mAMG$ in $O( m \diamG^2 (\tfrac{\Wmax}{\eps \Wmin})^2 \log n)$ time since $\mAMG \geq \Wmin$.
	\item \underline{Weighted vs unweighted diameter.} 	For simplicity, our runtime is written in terms of the unweighted diameter $d$. However, $\Wmax \diamG$ can be replaced by the weighted diameter of the graph with weights $w(e) -\Wmin$ which are translated to be all nonnegative.\footnote{This weighted diameter is a natural quantity since it is invariant under the simultaneous translation of all edge weights---a transformation which does not change the complexity of (additively approximating) $\MMC$. To get such bounds, the only change to our algorithms is to compute Single Source Shortest Paths using these translated weights (rather than unit weights), which can be done in near-linear time since they are nonnegative.} This yields tighter bounds since this weighted diameter is at most $\diamG$ times the weight range.
	\item \underline{Implications.} Our improved approximation algorithm for $\MMC$ immediately implies similarly improved algorithms for several related problems. For instance, the Min-GeoMean-Cycle problem---in which weights are strictly positive, and we seek a cycle $\sigma$ minimizing $(\prod_{e \in \sigma} w(e))^{1/|\sigma|}$---can be multiplicatively approximated by using our algorithms to additively approximate $\MMC$ with weights $\tilde{w}(e) := \log w(e)$. 
	Another immediate implication is the first near-linear time algorithm (again assuming moderate connectedness) for approximating fundamental quantities in max-plus spectral theory. Specifically, let $A$ be an $n \times n$ matrix with entries in $\Rmax = \R \cup \{-\infty\}$. It is known that the max-plus eigenvalues and the cycle-time vector of $A$ are characterizeable in terms of the Min-Mean-Cycles of the strongly connected components of the associated digraph $G = (\{1,\dots, n\}, \{(i,j) : A_{ij} \neq -\infty\})$, see, e.g.,~\citep{BapStaDri93,Gun94}.
	Thus, after topologically sorting the components of $G$ in linear time,
	we can compute both the max-plus spectrum and the cycle-time vector of $A$ to $\ell_{\infty}$ error $\eps$ in $\tilde{O}(m \diamG^2 (\tfrac{\Wmax}{\eps})^2)$ time, where $\Wmax := \max_{ij : A_{ij} \neq -\infty} |A_{ij}|$ and $\diamG$ denotes the diameter of $G$.
\end{description}

\subsection{Approach}\label{ssec:intro:tech}

\par In contrast to previous combinatorial approaches for $\MMC$, we tackle this discrete problem via \emph{continuous} optimization techniques. At a high level, we follow a standard template for approximation algorithms that consists of two steps: approximately solve a linear programming (LP) relaxation; then round the fractional solution to a vertex without worsening the LP cost by much. While this high-level template is standard, implementing it efficiently for $\MMC$ poses several obstacles. In particular, both steps require new specialized algorithms since out-of-the-box LP solvers and rounding algorithms are too slow for our desired runtime. Moreover, our goal of designing a memory-optimal algorithm restricts memory usage to being sublinear in the graph size, thereby precluding many natural approaches.

\par Our starting point is the classical LP relaxation of $\MMC$
\begin{align}
\min_{F \in \FE} \sum_{e \in E} F(e) w(e),
\tag{MMC-P}
\end{align}
where above the decision set $\FE$ is the polytope consisting of circulations on $G$ that are normalized to have unit total flow. Details on this LP are in the preliminaries section \S\ref{sec:prelim}.
\paragraph*{Step 1: optimization}
This is the main step of the algorithm---both conceptually and technically. In it, we find a near-optimal solution for~\eqref{MMC-P}. 
We do this by employing \emph{entropic regularization}, a celebrated technique for regularizing optimization problems over probability distributions. This is motivated by viewing the normalized circulations in $\FE$ as probability distributions on the edges of $G$ (see Remark~\ref{rem:FE}). The key insight is that entropically regularizing~\eqref{MMC-P} results in a convex optimization problem that corresponds to an associated \emph{Matrix Balancing} problem. This effectively reduces approximating~\eqref{MMC-P} to a problem for which near-linear time algorithms were recently developed~\citep{ZhuLiOliWig17,AltPar20bal,CohMadTsiVla17,OstRabYou16}.
In particular, we employ a randomized\footnote{This is the only source of randomness in our proposed algorithm.} version of Osborne's algorithm for Matrix Balancing which is practical and provably runs in near-linear time~\citep{AltPar20bal}.
A further benefit of our reduction is that Matrix Balancing can be performed in a memory-optimal way, yielding a fractional solution for~\eqref{MMC-P} that is compactly represented using $O(n)$ memory despite having $m$ nonzero entries.
See \S\ref{sec:bal} for details and for natural dual interpretations of the regularization and algorithm.
\paragraph*{Step 2: rounding} 
Step 1 outputs a near-feasible circulation (since Matrix Balancing can only be performed approximately) with near-optimal objective for~\eqref{MMC-P}. In this step, we compute from this a near-optimal cycle for $\MMC$. We perform this in two sub-steps.
\par First, we correct feasibility without changing much flow, thereby preserving near-optimality. We do this by re-routing flow from vertices with flow surplus to vertices with flow deficiency via short paths. While a na\"ive implementation of this requires $O(mn)$ time and $O(nd)$ memory, there is a simple trick that enables implementing this in near-linear time and in a memory-optimal way: route all these paths through an arbitrary vertex. Details in \S\ref{ssec:round:circ}. 
\par Second, we round the resulting near-optimal circulation (a fractional point in $\FE$) to a cycle (a vertex of $\FE$) while preserving the objective of~\eqref{MMC-P}. The Cycle-Cancelling algorithm~\citep{Sch03} does this by decomposing the circulation into a convex combination of cycles, and then outputting the best cycle. However, it has a prohibitive $O(mn)$ runtime. Since we can tolerate $\eps$ error, a Ford-Fulkerson-esque argument enables us to speed up this algorithm to near-linear time by simply running it on a quantization of the circulation. Details in \S\ref{ssec:round:extract}.

\subsection{Prior work}\label{ssec:intro:prev}

\subsubsection{Exact algorithms} 
There is an extensive literature on $\MMC$ algorithms; Table~\ref{tab:mmc-exact} summarizes the fastest known runtimes.
These runtimes are incomparable in that each is best for a certain parameter regime. The fastest algorithm for very large edge weights is the $O(mn)$ dynamic-programming algorithm of~\citep{Karp78}.\footnote{The algorithms of~\citep{Orl81,YouTarOrl91} have similar worst-case runtimes but better best-case and empirical runtimes.}
For more moderate weights (e.g., integers of polynomial size in $n$), the $O(m\sqrt{n} \log(n \Wmax ))$ scaling-based algorithm of~\citep{OrlAhu92} is faster.
Faster runtimes for certain parameter regimes are implicit from recent algorithmic developments for Single Source Shortest Paths ($\SSSP$). The connection is that $\SSSP$ algorithms can detect negative cycles, and $\MMC$ on an integer-weighted graph is reducible to detecting negative cycles on $O(\log (n \Wmax))$ graphs with modified edge weights~\citep{Law66}. This results in an $O(n^{\omega} \Wmax \log(n \Wmax))$ runtime which is faster for dense graphs with small weights~\citep{San05,YusZwi05}, and an
$m^{11/8+o(1)} \log^2 \Wmax$ runtime which is faster for sparse graphs with moderate weights~\citep{AxiMadVla20}.

\renewcommand\arraystretch{1.2}
\begin{table}[h]
	\centering
	\begin{tabular}{|c|c|c|}
		\hline
		\textbf{Author} & \textbf{Runtime} & \textbf{Memory} \\ \hline
		Karp (1978)~\citep{Karp78} & $O(mn)$ & $O(n^2)$ \\ \hline
		Orlin and Ahuja (1992)~\citep{OrlAhu92} & $\Otilde(m\sqrt{n})$ & $O(n)$ \\ \hline
		Sankowski (2005)~\citep{San05}, Yuster and Zwick (2005)~\citep{YusZwi05} & $\Otilde(n^{\omega})$ & $O(n^2)$  \\ \hline
		Axiotis et al. (2020)~\citep{AxiMadVla20} & $m^{11/8+o(1)}$ & $O(m)$ \\ \hline
	\end{tabular}
	\caption{
		\small
		Fastest runtimes for exact $\MMC$ computation. The memory reported is the additional storage beyond reading the input (see \S\ref{sssec:bal:final:mem}). For simplicity, here edge weights are in $\{-1,0,1\}$; see the main text for detailed dependence on $\Wmax$. 
	}
	\label{tab:mmc-exact}
\end{table}

\subsubsection{Approximation algorithms}

Table~\ref{tab:mmc-approx} lists the fastest approximation algorithms for $\MMC$. 
The fastest existing approximation algorithm is the
$\tilde{O}(n^{\omega} / \delta)$
algorithm of~\citep{CHKLR14} for approximating $\MMC$ to a $(1\plusminus \delta)$ multiplicative factor, in the special case of nonnegative integer weights. By taking $\delta = O(\eps/\Wmax)$, this can be converted into an $\plusminus \eps$ additive approximation algorithm with runtime
$\tilde{O}(n^{\omega} \Wmax / \eps)$.
This runtime is only faster than the exact algorithms of~\citep{San05,YusZwi05} by a factor of $\Otilde(1/\eps)$, which provides significant runtime gains only when the approximation accuracy $\eps$ is quite large.

\renewcommand\arraystretch{1.2}
\begin{table}[H]
	\centering
	\begin{tabular}{|c|c|c|}
		\hline
		\textbf{Author} & \textbf{Runtime} & \textbf{Memory} \\
		\hline
		Chatterjee et al. (2014)~\citep{CHKLR14} & $\tilde{O}(n^{\omega} / \eps )$ & $O(n^2)$ \\ \hline
		This paper (Theorem~\ref{thm:bal:orand}) & $\Otilde(m \diamG^2 / \eps^2)$ & $O(n)$ \\ \hline
	\end{tabular}
	\caption{
		\small
		Fastest runtimes for approximating $\MMC$ to $\eps$ additive accuracy. The memory reported is the additional storage beyond reading the input (see \S\ref{sssec:bal:final:mem}). For simplicity, here edge weights are in $\{-1,0,1\}$; see the main text for detailed dependence on $\Wmax$.
	}
	\label{tab:mmc-approx}
\end{table}

\par We also mention Howard's policy-iteration algorithm~\citep{Howard}. Although the fastest known theoretical runtime for it is slower\footnote{Namely, $O(mn^3\Wmax/\eps)$ for approximating $\MMC$ to $\eps$ additive accuracy if stopped early~\citep[Theorem 3.5]{Das04}.} than other algorithms, it is often used in practice
because its empirical performance significantly outperforms its theoretical runtime~\citep{CocCohGau98,Das04,DasIraGup99}. On the other hand, the practical runtime of Howard's algorithm is observed to be at least $\Omega(mn)$ rather than near-linear when run on ``difficult'' inputs~\citep{Das04,GGTW09}, see also Figure~\ref{fig:scal}.

\begin{remark}[Alternative approach]
	An alternative algorithm that uses the same rounding subroutine as $\Algbal$, but instead uses area-convexity regularization for the optimization subroutine, yields a slightly faster theoretical runtime of $\Otilde(m \diamG \Wmax / \eps)$. The tradeoff is that unlike $\Algbal$, this algorithm is not memory-optimal and performs poorly in practice. For details, see the extended version of this manuscript~\citep{AltPar20mmcarxiv}.
\end{remark}

\subsection{Simultaneous work}
After v1 of this manuscript was posted to arXiv, the paper~\citep{brand2020bipartite} appeared on arXiv (and has since appeared in FOCS~\citep{brand2020bipartitefocs}). That paper~\citep{brand2020bipartite} provides a breakthrough for solving a number of graph problems (including $\MMC$) in near-linear time for graphs that are sufficiently dense $m = \tilde{\Omega}(n^{1.5})$. We mention the tradeoffs between this $\MMC$ algorithm and ours. On one hand, their algorithm can compute exact solutions whereas ours can only compute approximations with moderate accuracy. On the other hand, (1) their algorithm relies on Laplacian solvers for which there is currently no practical implementation; (2) our algorithm is memory-optimal and uses $O(n)$ memory, compared to the $\Omega(m)$ used by theirs; and (3) our algorithm still has near-linear runtime for sparse graphs $m = o(n^{1.5})$ with small diameter.

\subsection{Roadmap}\label{ssec:intro:outline}

\S\ref{sec:prelim} recalls preliminaries. \S\ref{sec:fra} details the two steps in our approach---optimize and round---which we implement efficiently in \S\ref{sec:bal} and \S\ref{sec:round}, respectively. \S\ref{ssec:bal:final} puts these pieces together to conclude our algorithm. \S\ref{sec:experiments} provides preliminary numerical simulations.

%% file: sections/prelim.tex
\section{Preliminaries}\label{sec:prelim}

Throughout, we assume that $G$ is \textit{strongly connected}, i.e., that there is a directed path from every vertex to every other. This is without loss of generality since we can decompose a general graph $G$ into its strongly connected components in linear time~\citep{TarjanStronglyConnected}, and then solve $\MMC$ on $G$ by solving $\MMC$ on each component.

For simplicity, we assume each input edge weight is represented using an $\tilde{O}(1)$-bit number. This is essentially without loss of generality since after translating the weights and truncating them to $\plusminus \eps$ additive accuracy---which does not change the problem of additively approximating $\MMC$---all weights are representable using $O(\log (\Wmax / \eps)) = \tilde{O}(1)$-bit numbers.

In the sequel, we make use of a simple folklore algorithm for approximating the unweighted diameter $\diamG$ to within a factor of $2$ in $O(m)$ time. This algorithm, called $\Approxdiam$, runs Breadth First Search to and from some vertex $v$, and returns the sum of the maximum distance found to and from $v$. It is straightforward to show that the output $\tilde{d}$ satisfies $\diamG \leq \tilde{d} \leq 2 \diamG$.
Efficiently computing better approximations is an active research area, but this suffices for our purposes.

\subsection{Notation}\label{ssec:prelim:not}

Throughout, we reserve $G$ for the graph, $V$ for its vertex set, $E$ for its edge set, $w$ for its edge weights, $n = |V|$ for its number of vertices, $m = |E|$ for its number of edges, and $\diamG$ for its \textit{unweighted} diameter (i.e., the maximum over $u,v \in V$ of the shortest unweighted path from $u$ to $v$). For a positive integer $n$, we denote the set $\{1, \dots, n\}$ by $[n]$. 

\paragraph*{Linear algebraic notation}
Although this paper targets graph theoretic problems, it is often helpful---both for intuition and conciseness---to express things using linear algebraic notation. For a weighted digraph $G = (V,E,w)$, we write $W$ to denote the $n \times n$ matrix with $ij$-th entry $w(i,j)$ if $(i,j) \in E$, and $\infty$ otherwise. The support of a matrix $A$ is $\supp(A) := \{(i,j) \, : \, A_{ij} \neq 0\}$.
We write $\zero$ and $\bone$ to denote the all-zeros and all-ones vectors, respectively, in an ambient dimension clear from context (typically $\Rn$). 
For a vector $v \in \Rn$, we denote its $\ell_1$ norm by $\|v\|_1 := \sum_{i} |v_i|$, its $\ell_{\infty}$ norm by $\|v\|_{\infty} = \max_{i} |v_i|$, its entrywise exponentiation by $\exp[v]$, and its diagonalization by $\diag(v) \in \Rnn$. For a matrix $A$, we denote the $\ell_1$ norm of its vectorization by $\|A\|_1 := \sum_{ij} |A_{ij}|$, and its \emph{entrywise} exponentiation by $\exp[A]$.

\paragraph*{Flows and circulations}
A \textit{flow} on a digraph $G = (V,E)$ is a function $f : E \to \R_{\geq 0}$. Equivalently, in linear algebraic notation, this is a matrix $F \in \Rpnn$ with $\supp(F) \subseteq E$. The corresponding \textit{inflow}, \textit{outflow}, and \textit{netflow} for a vertex $i \in V$ are respectively $\sum_{(j,i) \in E} f(j,i)$, $\sum_{(i,j) \in E} f(i,j)$, and $\sum_{(j,i) \in E} f(j,i) - \sum_{(i,j) \in E} f(i,j)$; or in linear algebraic notation $(F^T\bone)_i$, $(F\bone)_i$, and $(F^T\bone - F\bone)_i$. A flow is \textit{balanced} at a vertex if that vertex has $0$ netflow. A \textit{circulation} is a flow that is balanced at each vertex. The \textit{total netflow imbalance} of a flow $F$ is denoted $\imb{F} := \|F \bone - F^T\bone\|_1$. A flow or circulation is \emph{normalized} if $\sum_{(i,j) \in E} f(i,j) = 1$.

\paragraph*{Probability distributions} 
The set of discrete distributions on $k$ atoms is associated with the $k$-simplex $\Delta_k := \{v \in \Rp^k : \sum_i v_i = 1\}$, the set of joint distributions on $V \times V$ with $\Delta_{n \times n}  := \{P \in \Rpnn  : \sum_{ij} P_{ij} = 1 \}$, and the set of distributions on $E$ with $\DE := \{P \in \Delta_{n \times n}  :  \supp(P) \subseteq E \}$.

\subsection{LP relaxations of Min-Mean-Cycle}\label{ssec:prelim:lp}

Here we recall the classical primal/dual pair of LP relaxations of $\MMC$.
Consider a weighted digraph $G = (V,E,w)$. Associate to each cycle $\sigma$ an $n \times n$ matrix $F_{\sigma}$ with $ij$-th entry equal to $1/|\sigma|$ if $(i,j) \in \sigma$, and $0$ otherwise. 
Then $\MMC$ can be formulated as $
\mAMG = \min_{\text{cycle } \sigma} \langle F_{\sigma}, W \rangle$, where the inner product $\langle F_{\sigma},W \rangle := \sum_{(i,j) \in E} (F_{\sigma})_{ij} W_{ij}$ ranges over the edges of $G$. The LP relaxation of this discrete problem is
\begin{align}
\min_{F \in \FE} \langle F, W \rangle,
\tag{MMC-P}
\end{align}
where $\FE$ is the convex hull of $\{ F_{\sigma} : \sigma \text{ cycle} \}$. It is well-known (e.g.,~\citep[Problem 5.47]{AhuMagOrl88}) that
\[
\FE
=
\{F  \in \DE : F\bone = F^T\bone\}.
\]
\vspace{-0.7cm}
\begin{remark}[Interpretations of $\FE$]\label{rem:FE}
	From a graph theoretic perspective, $\FE$ is the set of \textit{normalized circulations on $G$}; and from a probabilistic perspective, $\FE$ is the set of \emph{joint distributions on the edge set $E \subseteq V \times V$ with identical marginal distributions.} 
	There are also natural interpretations of the $\ell_1$ distance $\|F\bone - F^T\bone\|_1$ of a matrix $F \in \DE$ from $\FE$: from a graph theoretic perspective, it
	is the \emph{total netflow imbalance};
	and from a probabilistic perspective, it is (two times) the \emph{total variation distance between the marginals}.
\end{remark}

Throughout, we call~\eqref{MMC-P} the primal LP relaxation. We refer to the dual of~\eqref{MMC-P} as the dual LP relaxation. This is the LP $\max_{p \in \R^n, \lambda \in \R \, : \, \lambda \leq W_{ij} + p_i - p_j, \, \forall (i,j) \in E} \lambda$, but in the sequel it is helpful to re-write it in the following saddle-point form:
\begin{align}
\max_{p \in \R^n} \min_{(i,j) \in E} W_{ij} + p_i - p_j.
\tag{MMC-D}
\end{align}

%% file: sections/framework.tex
\section{Algorithmic framework}\label{sec:fra}

Here we detail the algorithmic framework we use for approximating $\MMC$. As overviewed in \S\ref{ssec:intro:tech}, the framework consists of two steps: approximately solve the LP relaxation~\eqref{MMC-P}, and then round this fractional solution to a vertex with nearly as good value for~\eqref{MMC-P}. While the optimization step is sufficient for estimating the \textit{value} $\mAMG$ of $\MMC$, the rounding step yields a feasible \textit{solution} (i.e., a cycle).

\par Algorithm~\ref{alg:fra} summarizes the accuracy required of each step. Note that the optimization step produces a near-optimal solution that is not necessarily feasible, but rather \emph{near-feasible} in that we allow a slightly imbalanced netflow $\imbP = \|P\bone - P^T\bone\|_1$ up to some $\delta > 0$; in the sequel, we take $\delta = \Theta(\eps/(\Wmax\diamG))$. Our rounding step accounts for this near-feasibility.

\begin{algorithm}[H]
	\caption{Algorithmic framework for approximating $\MMC$.}
	\hspace*{\algorithmicindent} \textbf{Input:} Weighted digraph $G = (V,E,w)$, accuracy $\eps > 0$ \\
	\hspace*{\algorithmicindent} \textbf{Output:} Cycle $\sigma$ in $G$ satisfying $\wbar(\sigma) \leq \mAM(G) + \eps$
	\begin{algorithmic}[1]
		\Statex \textbackslash\textbackslash$\;$ Optimization step: compute near-feasible, near-optimal solution $P$ for~\eqref{MMC-P}
		\State Find matrix $P \in \DE$ satisfying $\imbP \leq \delta$ and $\langle P, W \rangle \leq \mAMG + 
		\frac{\eps}{2}$ \label{line:ammc:opt}
		\Statex
		\Statex \textbackslash\textbackslash$\;$ Rounding step: round $P$ to a vertex of $\FE$ with nearly as good cost for~\eqref{MMC-P}
		\State Find cycle $\sigma$ satisfying $\wbar(\sigma) \leq \langle P, W \rangle
		+ \frac{\eps}{4} + \frac{\eps \imbP}{4\delta}
		$
	\end{algorithmic}
	\label{alg:fra}
\end{algorithm}

\begin{obs}[Approximation guarantee for Algorithm~\ref{alg:fra}]\label{obs:fra}
	Given any weighted digraph $G$ and any accuracy $\eps > 0$, Algorithm~\ref{alg:fra} outputs a cycle $\sigma$ in $G$ satisfying $\wbar(\sigma) \leq \mAM(G) + \eps$.
\end{obs}

The proof is immediate by definition of the algorithmic framework. The obstacle is how to efficiently implement the two steps. This is shown in the following two sections.

%% file: sections/bal.tex
\section{Efficient optimization of the LP relaxation}\label{sec:bal}
Here, we use Matrix Balancing to efficiently implement the optimization in the framework described in \S\ref{sec:fra}. Below,~\S\ref{ssec:bal:connection} describes the connections between $\MMC$ and Matrix Balancing, and~\S\ref{ssec:bal:opt} makes this algorithmic.
\par Some preliminary definitions for this section. A matrix $A \in \Rpnn$ is \emph{balanced} if $A\bone = A^T\bone$. The \textit{Matrix Balancing problem} for input $K \in \Rpnn$ is to find a positive diagonal matrix $D$ (if one exists) such that $A = DKD^{-1}$ is balanced.\footnote{
	Technically, this is the problem of Matrix Balancing in the $\ell_1$ norm, since the goal is to match the $\ell_1$ norm of the rows and columns of $A$. However, we simply call this task ``Matrix Balancing'' because every instance of Matrix Balancing in this paper is in the setting of the $\ell_1$ norm.
} $K$ is \emph{balanceable} if such a solution $D$ exists (see Remark~\ref{rem:balanceable}). The notion of \textit{approximate Matrix Balancing} is introduced later in \S\ref{ssec:bal:opt}.

\subsection{Connection to Matrix Balancing}\label{ssec:bal:connection}

The key connection is that appropriately regularizing the LP relaxation of $\MMC$ results in a convex optimization problem that is equivalent to an associated Matrix Balancing problem. This regularization can be equivalently performed on either the primal or dual LP (see Table~\ref{tab:bal-reg}); we describe both perspectives as they give complementary insights. We note that while these regularized problems are well-known to be connected to Matrix Balancing (e.g.,~\citep{Elf80,KalKhaSho97}), the relation of $\MMC$ to these regularized problems and Matrix Balancing is, to our knowledge, not previously known.

\subsubsection{Primal regularization}

\begin{table}
	\resizebox{1\textwidth}{!}{\begin{minipage}{\textwidth}
			\centering
			\begin{tabular}{|c|c|c|}
				\hline
				& \textbf{Primal}  & \textbf{Dual}  \\ 
				\hline
				\textbf{Min-Mean-Cycle}
				& $\parbox{6cm}{\begin{equation} \min_{F \in \FE} \langle F, W \rangle  \tag{MMC-P} \label{MMC-P} \end{equation}}$ 
				& $\parbox{6cm}{\begin{equation} \max_{p \in \Rn} \min_{ij} W_{ij} + p_i - p_j \tag{MMC-D} \label{MMC-D} \end{equation}}$
				\\ 
				\hline
				\textbf{Matrix Balancing} 
				& $\parbox{6cm}{\begin{equation} \min_{F \in \FE} \langle F, W \rangle - \eta^{-1} H(F) \tag{MB-P} \label{MB-P} \end{equation}}$
				& $\parbox{6cm}{\begin{equation} \max_{p \in \Rn} \smineta_{ij} W_{ij} + p_i - p_j \tag{MB-D} \label{MB-D} \end{equation}}$
				\\ \hline
			\end{tabular}
			\caption{
				\small
				Primal/dual LP relaxations of $\MMC$ (top),
				and our proposed regularizations (bottom). 
				The regularized problems~\eqref{MB-P} and~\eqref{MB-D} are (essentially) dual convex programs, with (essentially) unique solutions corresponding to balancing the matrix $K = \exp[-\eta W]$.
			}
			\label{tab:bal-reg}
	\end{minipage} }
\end{table}

In the primal, we employ \textit{entropic regularization}: we subtract $\eta^{-1}$ times the Shannon entropy $H(F)$ from the objective in the primal LP relaxation~\eqref{MMC-P}. Recall that the Shannon entropy of a discrete distribution $p$ is $H(p) := - \sum_i p_i \log p_i$,
where we adopt the standard convention that $0 \log 0 = 0$. Note that this regularization results in a strictly convex optimization problem by strict concavity of the entropy. This regularization is motivated by the Max-Entropy principle; indeed, recall from Remark~\ref{rem:FE} the interpretation of~\eqref{MMC-P} as an optimization over probability distributions. The choice of the regularization parameter $\eta$ is discussed in Remark~\ref{rem:bal:eta} below, and is based on balancing the fact that~\eqref{MB-P} is ``more convex'' and thus easier to solve for small $\eta$, while its fidelity to the original problem~\eqref{MMC-P} improves for large $\eta$ due to the following basic bound.

\begin{lemma}[Entropy bound]\label{lem:ent}
	For any probability distribution $p \in \Delta_K$ with support size $k := |\{i \in [K] \, : \, p_i \neq 0 \}| \leq K$, we have $0 \leq H(p) \leq \log k$.
\end{lemma}

\subsubsection{Dual regularization} 

\par In the dual, we employ \textit{softmin smoothing}: we re-write the dual LP relaxation as the max-min saddle-point problem~\eqref{MMC-D}, and then replace the inner min by a smooth approximation $\smineta$, which is defined for a parameter $\eta > 0$ by 
\[
\smineta_{i \in [k]} a_i
:=
-\frac{1}{\eta} \log \left( \sum_{i=1}^k e^{-\eta a_i} \right),
\] 
where we adopt the standard convention $e^{-\infty} = 0$ to extend this notation to $a_i \in \R \cup \{+\infty \}$. Note that this regularization results in a concave optimization problem by concavity of the softmin function---in fact, strictly concave on the orthogonal complement of the subspace spanned by $\bone$. A similar discussion as for the primal regularization applies about the choice of regularization parameter $\eta$, except that here the fidelity of the regularized problem to the original unregularized problem is based on the following basic bound. 

\begin{lemma}[Softmin approximation bound]\label{lem:smin}
	For any $a_1, \dots, a_k \in \R \cup \{+\infty\}$ and $\eta > 0$,
	\[
	0 
	\leq
	\min_{i \in [k]} a_i - 	\smineta_{i \in [k]} a_i
	\leq  \frac{\log k}{\eta}
	\]
\end{lemma}

The regularized optimization problem~\eqref{MB-D} is given in Table~\ref{tab:bal-reg}. Expanding the softmin and re-parameterizing $x := -\eta p$ gives the more convenient equivalent form:
\begin{align}
- \frac{1}{\eta}\, \min_{x \in \R^n}
\log 
\left( \sum_{ij} e^{x_i - x_j} K_{ij} \right),
\tag{MB-D'}
\label{MB-D'}
\end{align}
where $K := \exp[-\eta W]$ denotes the \textit{entrywise} exponentiated matrix with entries $K_{ij} = e^{-\eta W_{ij}}$.

\subsubsection{Connections and remarks}

Not only are~\eqref{MB-P} and~\eqref{MB-D} both convex optimization problems, but also they are convex duals\footnote{Formally, this requires equivalently re-writing~\eqref{MB-D} in constrained form.} satisfying strong duality. The optimality conditions clarify the connection between these problems and Matrix Balancing: the (unique) solution of~\eqref{MB-P} corresponds to the (unique) balancing of $K$ modulo normalization, and the solutions of~\eqref{MB-D'}  (unique up to translation by $\bone$) correspond to the diagonal balancing matrices (unique up to a constant factor). This is formally stated as follows.

\begin{lemma}[Optimality conditions for~\eqref{MB-P} and~\eqref{MB-D'}]\label{lem:bal:kkt}
	Let $G = (V,E,w)$ be strongly connected and $\eta > 0$. Then:
	\begin{itemize}
		\item [(1)] $F \in \FE$ and $x \in \Rn$ are optimal solutions for~\eqref{MB-P} and~\eqref{MB-D'}, respectively, if and only if $F = A / \sum_{ij} A_{ij}$, where $A = \diag(e^x) K \diag(e^{-x})$.
		\item [(2)] \eqref{MB-P} has a unique solution. The solutions to~\eqref{MB-D'} are unique up to translation by $\bone$.
	\end{itemize} 
\end{lemma}

A similar result can be found in~\citep[Theorem 1]{KalKhaSho97}, although the focus there is on the dual regularized problem. For completeness, we provide a short proof here that highlights the primal regularized problem and the convex duality.

\begin{proof}
	Dualize the affine constraint $F\bone = F^T\bone$ in~\eqref{MB-P} via the penalty $p^T(F\bone - F^T\bone) = \sum_{(i,j) \in E}F_{ij}(p_i - p_j)$, where $p \in \Rn$ is the associated Lagrange multiplier. This results in the minimax problem
	\begin{align}
		\min_{F \in \DE} \max_{p \in \Rn} \sum_{(i,j) \in E} F_{ij} (W_{ij} + p_i - p_j + \eta^{-1} \log F_{ij})
		\label{eq:bal-reg:minmax}
	\end{align}
	By Sion's Minimax Theorem~\citep{Sio58}, this equals the maximin problem
	\begin{align}
		\max_{p \in \Rn} \min_{F \in \DE}  \sum_{(i,j) \in E} F_{ij} (W_{ij} + p_i - p_j + \eta^{-1} \log F_{ij})
		\label{eq:bal-reg:maxmin}
	\end{align}
	The inner minimization problem can now be solved explicitly. A standard Lagrange multiplier calculation shows that at optimality, $F$ is the matrix with $ij$-th entry equal to
	\begin{align}
		F_{ij} = c e^{-\eta (W_{ij} + p_i - p_j)}
		\tag{BAL-OPT}
		\label{BAL-OPT}
	\end{align}
	where $c = 1/(\sum_{(i,j) \in E} e^{-\eta (W_{ij} + p_i - p_j)})$ is the normalizing constant. (Note that if $(i,j) \notin E$, then $W_{ij} = \infty$, so $F_{ij} = e^{-\eta W_{ij}} = 0$.) Plugging~\eqref{BAL-OPT} into~\eqref{eq:bal-reg:maxmin} and simplifying yields
	\begin{align}
		\max_{p \in \Rn} -\eta^{-1} \log \left(\sum_{(i,j) \in E} e^{-\eta (W_{ij} + p_i - p_j)} \right),
	\end{align}
	which is precisely~\eqref{MB-D}. This establishes strong duality. Item (1) then follows from the optimality condition established above in~\eqref{BAL-OPT}. 
	\par For item (2), strict concavity of entropy implies that~\eqref{MB-P} has a unique optimal solution. This combined with the optimality condition in item (1) implies that $\diag(e^x)K\diag(e^{-x})$ is invariant among optimal solutions $x$ of~\eqref{MB-D'}. Thus if $x$ and $y$ are both solutions, then $x_i - x_j = y_i - y_j$ for all edges $(i,j) \in E$. It follows that in each strongly connected component of $G$, the difference $x_i - y_i$ is constant over all vertices $i$. Since $G$ is assumed strongly connected, $x$ and $y$ are equal up to an additive shift of $\bone$.
\end{proof}

\par The strongly connected assumption in Lemma~\ref{lem:bal:kkt} is important for balanceability:

\begin{remark}[Balanceability for $\MMC$]\label{rem:balanceable}
	$K \in \Rpnn$ is \emph{balanceable} if and only if $K$ is \emph{irreducible}---i.e., the graph $G_K = ([n], \supp(K))$ is strongly connected~\citep{Osborne60}. Thus, in our $\MMC$ application, $K = \exp[-\eta W]$ is balanceable since $G = G_K$ is strongly connnected (see \S\ref{sec:prelim}). Furthermore, balanceability is necessary and sufficient for uniqueness (modulo translation) of the solutions to the dual regularized problem~\eqref{MB-D'}, essentially because balanceability can be shown to be equivalent to strict concavity of the dual regularized problem~\eqref{MB-D'} on the orthogonal complement of the subspace spanned by $\bone$.
\end{remark}

We conclude this discussion with two remarks about the regularization parameter $\eta$.

\begin{figure}
	\centering
	\includegraphics[width=\linewidth]{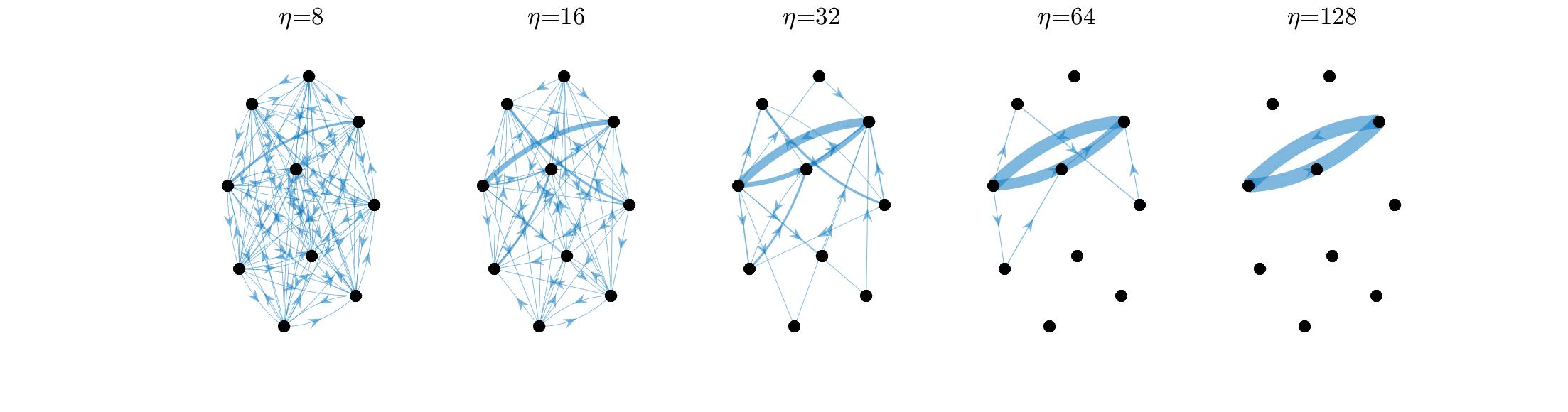}
	\caption{\small
		Effect of entropic regularization on the sparsity of the optimal solution $F^{\eta}$ of~\eqref{MB-P}. Plotted here is $F^{\eta}$ for varying regularization parameter $\eta$,
		 where edge $(i,j)$ is drawn with width proportional to $F_{ij}^{\eta}$, and dropped if it has sufficiently small mass.
	}
	\label{fig:eta:dense}
\end{figure}

\begin{remark}[Effect of regularizing $\MMC$]\label{rem:bal:eta:sparse}
	The solution $\Feta$ is readily characterized in the limit as the regularization dominates ($\eta \to 0$) or vanishes ($\eta \to \infty$): $\lim_{\eta \to 0} \Feta$ is the max-entropy element of $\FE$, and $\lim_{\eta \to \infty} \Feta$ is the max-entropy solution among optimal solutions for~\eqref{MMC-P}.\footnote{This is in analog to entropic Optimal Transport~\citep[Proposition 4.1]{PeyCut17}, and can be proved similarly.} For every finite $\eta$, the solution $F^{\eta}$ is dense in that $F_{ij}^{\eta} > 0$ for every edge $(i,j)$. However, as $\eta$ increases (i.e., the regularization decreases), $F^{\eta}$ concentrates on edges belonging to Min-Mean-Cycle(s); see Figure~\ref{fig:eta:dense} for an illustration.
\end{remark}

\begin{remark}[Tradeoff for regularizing $\MMC$]\label{rem:bal:eta}
	There is a natural algorithmic tradeoff for choosing $\eta$: roughly, more regularization makes $K = \exp[-\eta W]$ easier to balance, while less regularization ensures fidelity of the regularized problems to the original LPs. Therefore, we take $\eta$ as small as possible such that solving the regularized problems yields an $O(\eps)$ optimal solution for the original LPs (and thus $\MMC$). 
	A simple argument---either bounding the primal entropy regularization by $\eta^{-1} \log m$ using Lemma~\ref{lem:ent}, or bounding the dual softmin approximation error by $\eta^{-1} \log m$ using Lemma~\ref{lem:smin}---shows that $\eta = O(\eps^{-1} \log m)$ suffices.
\end{remark}

\subsection{Optimization via Matrix Balancing}\label{ssec:bal:opt}

We now make the connections in \S\ref{ssec:bal:connection} algorithmic by reducing the optimization step in the algorithmic framework described in \S\ref{sec:fra}, to Matrix Balancing. Although Matrix Balancing is difficult to perform \emph{exactly}, we show that performing it \emph{approximately} suffices.

\begin{defin}[Approximate Matrix Balancing]\label{def:abal}
	A nonnegative matrix $A$ is \emph{$\delta$-balanced} if
	\begin{align}
	\frac{\|A\bone - A^T\bone\|_1}{\sum_{ij} A_{ij}}
	\leq \delta.
	\label{eq:bal:approx}
	\end{align}
	The \emph{approximate Matrix Balancing problem} for $K \in \Rpnn$ and $\delta > 0$ is to find a positive diagonal matrix $D$ such that $A := DKD^{-1}$ is $\delta$-balanced and satisfies $\sum_{ij}A_{ij} \leq \sum_{ij} K_{ij}$.\footnote{The second condition $\sum_{ij}A_{ij} \leq \sum_{ij} K_{ij}$ is only for technical purposes (it ensures conditioning bounds, see Lemma~\ref{lem:bal:R}) and is a mild requirement since all natural balancing algorithms satisfy it. Indeed, balancing $K$ is equivalent to minimizing $\sum_{ij} A_{ij}$ (Lemma~\ref{lem:bal:kkt}), and $\sum_{ij} K_{ij}$ is the value of $\sum_{ij} A_{ij}$ without any balancing.
	}
\end{defin}

We now state the main result of this section: a reduction from the optimization step in the algorithmic framework described in \S\ref{sec:fra}, to approximately balancing the matrix $K = \exp[-\eta W]$ to accuracy $\delta = \Theta(\eps/(\Wmax \diamG))$, where $\eta  = \Theta((\log m)/\eps)$. The upshot is that this allows us to leverage known near-linear time algorithms for approximate Matrix Balancing.

\begin{theorem}[Efficient optimization via Matrix Balancing]\label{thm:bal:opt}
	Let $G = (V,E,w)$ be strongly connected, $\eta = (2.5 \log m) / \eps$, and $\delta \leq \eps/(16 \Wmax \diamG)$. Let $x \in \R^n$ be such that $\diag(e^x)$ solves the $\delta$-approximate Matrix Balancing problem on $K = \exp[-\eta W]$, and denote $A := \diag(e^x)K\diag(e^{-x})$. Then $P = A/(\sum_{ij} A_{ij})$ satisfies $P \in \Delta_E$, $\imbP \leq \delta$, and $\langle P, W \rangle \leq \mAMG + \eps/2$.
\end{theorem}

It is clear by construction that $P \in \DE$ and $\imbP \leq \delta$; the near-optimality $\langle P, W \rangle \leq \mAMG + \eps/2$ is what requires proof.  
The intuition is as follows. Since $P$ is approximately balanced, the (nearly feasible) pair of primal-dual solutions $(P,x)$ nearly satisfies the optimality conditions in Lemma~\ref{lem:bal:kkt}, and thus $P$ is nearly optimal for~\eqref{MB-P}. Since~\eqref{MB-P} is pointwise close to the primal LP relaxation~\eqref{MMC-P} (since the regularization is small by Lemma~\ref{lem:ent}), therefore $P$ is also nearly optimal for the original optimization problem~\eqref{MMC-P}.

\par To formalize this intuition we require three lemmas. First, we compute the gap between objectives for a certain family of primal-dual ``solution'' pairs for~\eqref{MB-P} and~\eqref{MB-D'} inspired by the optimality conditions in Lemma~\ref{lem:bal:kkt}. Note that the primal solution may not be feasible since it may not be balanced---in fact, Lemma~\ref{lem:duality-gap} shows that this imbalance controls this gap.

\begin{lemma}[Duality gap]\label{lem:duality-gap}
	Let $\eta > 0$ and $x \in \R^n$. Define $K = \exp[-\eta W]$, $A = \diag(e^x)K\diag(e^{-x})$, and $P = A/(\sum_{ij} A_{ij})$. Then 
	$(\langle P,W\rangle-\eta^{-1}H(P)) - (-\eta^{-1}\log\sum_{ij}A_{ij}) =  \eta^{-1} x^T(P\bone - P^T\bone)$.
\end{lemma}
\begin{proof}
	Straightforward calculation.
\end{proof}

\par The second lemma shows that the dual balancing objective gives a lower bound on $\MMC$. This amounts to the pointwise nonnegativity of our regularizations of the LP relaxations.

\begin{lemma}[Lower bound on $\MMC$ via balancing]\label{lem:bal:lb}
	Consider any $\eta > 0$ and $x \in \Rn$. Let $K = \exp[-\eta W]$ and $A = \diag(e^x)K\diag(e^{-x})$. Then $-\eta^{-1} \log \sum_{ij} A_{ij} \leq \mAM(G)$.
\end{lemma}
\begin{proof}
	Let $p = \eta x$. By Lemma~\ref{lem:smin}, $-\eta^{-1} \log \sum_{ij} A_{ij} = \smin_{\eta,\,(i,j) \in E } W_{ij} + p_i - p_j \leq \min_{(i,j) \in E} W_{ij} + p_i - p_j$. By feasibility of $p$ for the dual LP relaxation~\eqref{MMC-D}, this is at most $\mAM(G$).
\end{proof}

\par The third lemma is a standard conditioning bound (e.g.,~\citep[Lemma 3.5]{AltPar20bal}) for nontrivial balancings, i.e., $x \in \Rn$ with objective for~\eqref{MB-D'} no worse than $\zero$. Below, let $\condK := \tfrac{\sum_{ij} K_{ij}}{\min_{ij \in \supp(K)} K_{ij}} $.

\begin{lemma}[Conditioning of nontrivial balancings]\label{lem:bal:R}
	Let $K \in \Rpnn$ be balanceable and $G = ([n], \supp(K))$. If $x \in \R^n$ satisfies $\sum_{ij} e^{x_i - x_j} K_{ij} \leq \sum_{ij} K_{ij}$, then
	$\max_i x_i - \min_i x_i \leq \diamG
	\log \condK
	$.
\end{lemma}

Note that in $\Algbal$, we have $K = \exp[-\eta W]$ and $\eta = O((\log m)/\eps)$, and thus
\begin{align}
\log \condK
\leq \log \frac{m \exp(\eta \Wmax)}{\exp(-\eta \Wmax)}
= \log m + 2 \eta \Wmax
\label{eq:bal:condK}
\end{align}
which is of size $O\left((\Wmax/\eps) \log m \right)$. We are now ready to prove Theorem~\ref{thm:bal:opt}.

\begin{proof}[Proof of  Theorem~\ref{thm:bal:opt}]
	Rearranging the inequality in Lemma~\ref{lem:duality-gap} yields 
	\[
	\langle P, W \rangle = -\eta^{-1} \log \sum_{ij}A_{ij} + \eta^{-1} H(P) + \eta^{-1} x^T(P\bone - P^T\bone).
	\]
	We show the right hand side is at most $\mAM(G) + \eps/2$. The first term is at most $-\eta^{-1} \log \sum_{ij}A_{ij} \leq \mAM(G)$ by Lemma~\ref{lem:bal:lb}. The second term is at most $\eta^{-1}H(P) \leq \eta^{-1}\log m = 2\eps/5$ by Lemma~\ref{lem:ent} and the choice of $\eta$. 
	Finally, the third term is at most
	\begin{align*}
	\frac{1}{\eta} x^T(P\bone - P^T\bone)
	\leq \frac{1}{2\eta}  (\max_i x_i - \min_i x_i)\, \|P\bone - P^T\bone\|_1
	\leq \frac{\delta \diamG }{2\eta} \log \condK
	\leq \frac{\eps}{10},
	\end{align*}
	where above the first inequality is by applying H\"older's inequality after possibly re-centering $x$ (since $x \mapsto x^T (P\bone - P^T\bone)$ is invariant under adding multiples of the all-ones vector $\bone$ to $x$); the second inequality is by Lemma~\ref{lem:bal:R} and the construction of $P$ by re-normalizing a $\delta$-balanced matrix; and the final inequality is by the conditioning bound~\eqref{eq:bal:condK}, the choice of $\eta$, and the bound $\eps \leq 2\Wmax$ (which may be assumed otherwise every cycle is $\eps$-suboptimal).
\end{proof}

%% file: sections/round.tex
\section{Efficient rounding of the LP relaxation}\label{sec:round}

Here we present an efficient implementation of the rounding step in the algorithmic framework described in \S\ref{sec:fra}.

\begin{theorem}[Efficient rounding]\label{thm:round}
	There is an algorithm (namely, $\Quantroundtocirc$ in \S\ref{ssec:round:circ}
	followed by $\Extractcycle$ in \S\ref{ssec:round:extract}) that, given $G = (V,E,w)$, a normalized flow $P \in \DE$ with netflow imbalance $\imbP \leq 1/\diamG$, and an accuracy $\eps > 0$, takes $O(m \diamG \Wmax / \eps)$ time to output a cycle $\sigma$ in $G$ satisfying
	\[
	\wbar(\sigma) \leq \langle P, W \rangle
	+ \frac{\eps}{4}
	+ 4 \Wmax \diamG \imbP.
	\]
	In particular, if $\imbP \leq \eps/(16\Wmax\diamG)$, then $\wbar(\sigma)
	\leq
	\langle P, W \rangle
	+ \eps/2$.
\end{theorem}

Furthermore, this algorithm can be implemented using only $O(n)$ additional memory. But since this modification is a minor extension, we defer it to Appendix~\ref{app:round:mem} for ease of exposition.

\par We perform the rounding in two steps. First, $\Quantroundtocirc$ rounds the near-circulation $P$ to a circulation $F \in \FE$ such that (i) little flow is adjusted, and (ii) $F$ is $\gamma$-quantized\footnote{We say a matrix is $\gamma$-quantized if each entry is an integer multiple of $\gamma$.} for an appropriately chosen scalar $\gamma$. Property (i) ensures that the cost is approximately preserved, and property (ii) enables the efficient implementation of the second step. Second, $\Extractcycle$ rounds $F \in \FE$ to a vertex while preserving the cost. The formal guarantees are as follows.

\begin{lemma}[Guarantee for $\Quantroundtocirc$]\label{lem:quantroundtocirc}
Given $G = (V,E,w)$, $P \in \DE$ satisfying $\imbP \leq 1/\diamG$, and $\eps > 0$, $\Quantroundtocirc$ takes $O(m + n \diamG)$ time to output $F \in \FE$ such that $F$ is $\gamma$-quantized for $\gamma = \Omega (\eps/(m \diamG \Wmax))$, and
	\begin{align}
	\|F - P\|_1 \leq 4 \diamG \imbP + \frac{\eps}{4\Wmax}.
	\label{eq:quantroundtocirc}
	\end{align}
\end{lemma}

\begin{lemma}[Guarantee for $\Extractcycle$]\label{lem:extractcycle}
	Given $G = (V,E,w)$ and a $\gamma$-quantized $F \in \FE$, $\Extractcycle$ takes $O(m + \gamma^{-1})$ time to output a cycle $\sigma$ satisfying $\wbar(\sigma) \leq \langle W, F \rangle$.
\end{lemma}

The proof of Theorem~\ref{thm:round} is immediate from these two lemmas. 

\begin{proof}[Proof of Theorem~\ref{thm:round}]
	The runtime follows from Lemmas~\ref{lem:quantroundtocirc} and~\ref{lem:extractcycle}. Let $F$ be the output of $\Quantroundtocirc$.
	By Lemma~\ref{lem:extractcycle}, $\wbar(\sigma) \leq \langle F, W\rangle = \langle P, W \rangle + \langle F-P, W \rangle$. By 
	H\"older's inequality and 
	Lemma~\ref{lem:quantroundtocirc}, $\langle F-P, W\rangle \leq \Wmax \|F-P\|_1 \leq  4 \Wmax \diamG \imbP + \eps/4$.
\end{proof}

\S\ref{ssec:round:circ} and \S\ref{ssec:round:extract} respectively detail these subroutines $\Quantroundtocirc$ and $\Extractcycle$, and prove their respective guarantees Lemmas~\ref{lem:quantroundtocirc} and~\ref{lem:extractcycle}.

\subsection{Rounding to the circulation polytope}\label{ssec:round:circ}

Here we describe the algorithm $\Quantroundtocirc$ and prove Lemma~\ref{lem:quantroundtocirc}. 
Let us first ignore quantization: given $G$ and a normalized flow $P \in \DE$, how to efficiently compute a normalized circulation $F \in \FE$ such that the adjusted flow $\|F-P\|_1$ is small compared to the total netflow imbalance $\imbP = \|P \bone - P^T\bone\|_1$? 
Since this does not require edge weights, we may presently think of $G$ as unweighted. 
\par A simple approach is: until all vertices have balanced flow, push flow from any vertex $i$ with negative netflow to any vertex $j$ with positive netflow along the shortest path in $G$ until $i$ or $j$ is balanced. After a normalization at the end, this produces an $F \in \FE$ satisfying\footnote{This follows from essentially the same argument as in the proof of Lemma~\ref{lem:roundtocirc}.} 
\begin{align}
\|F-P\|_1 = O( \diamG \imbP ).
\label{eq:roundtocirc:simple}
\end{align}
While this ratio $\|F-P\|_1 / \imbP$ is optimally small in the worst-case, the runtime is a prohibitive $\Theta(mn)$. The bottleneck is $\Theta(n)$ shortest path computations, each taking $\Theta(m)$ time. 
\par A simple trick for speeding this up while maintaining~\eqref{eq:roundtocirc:simple} is to use cheap estimates of the shortest paths that are of length at most $2\diamG$. Specifically, choose any vertex $v \in V$, and route all paths used in the flow-rebalancing through $v$ using the shortest path to/from $v$. See Algorithm~\ref{alg:roundtocirc} for pseudocode. Note that computing all shortest paths to/from $v$ (line~\ref{line:roundtocirc:sssp} of $\Roundtocirc$) takes $O(m)$ time by running two Breadth First Searches~\citep[\S6.2]{Sch03}.

\begin{algorithm}
	\caption{
		$\Roundtocirc$:
		efficiently rounds to $\FE$ without adjusting much flow.
	}
	\hspace*{\algorithmicindent} \textbf{Input:} Digraph $G = (V,E)$, normalized flow $P \in \DE$ \\
	\hspace*{\algorithmicindent} \textbf{Output}: Normalized circulation $F \in \FE$ satisfying~\eqref{eq:roundtocirc}
	\begin{algorithmic}[1]
		\State Choose $v \in V$, 
		compute shortest paths to and from $v$
		\label{line:roundtocirc:sssp}
		\State $Q \gets P$, $\imbQ \gets Q^T\bone - Q\bone$ \label{line:roundtocirc:imbQ}
		\Comment{Initial imbalance}
		\While{$\imbQ \neq 0$}
		\State Choose any vertices $i$ and $j$ with $\imbind{i}{Q} > 0$ and $\imbind{j}{Q} < 0$
		\State $\delta_{ij} \gets \min(\imbind{i}{Q}, -\imbind{j}{Q})$
		\State Add $\delta_{ij}$ in $Q$ to each edge on paths found in line~\ref{line:roundtocirc:sssp} from $i$ to $v$ to $j$ \label{line:roundtocirc:pathij}
		\Comment{Push flow}
		\State $\imbind{i}{Q} \gets \imbind{i}{Q} - \delta_{ij}$, $\imbind{j}{Q} \gets \imbind{j}{Q} + \delta_{ij}$
		\Comment{Update imbalance}
		\EndWhile
		\State \Return{$F \gets Q / \sum_{ij} Q_{ij}$}
	\end{algorithmic}
	\label{alg:roundtocirc}
\end{algorithm}

\begin{lemma}[Guarantee for $\Roundtocirc$]\label{lem:roundtocirc}
	Given a strongly connected digraph $G = (V,E)$ and a matrix $P \in \DE$,
	$\Roundtocirc$ takes $O(m + n \diamG)$ time to output $F \in \FE$ satisfying
	\begin{align}
	\|F - P\|_1 \leq 2 \diamG \imbP.
	\label{eq:roundtocirc}
	\end{align}
\end{lemma}
\begin{proof}
	All steps besides the while loop take $O(m)$ time. For this loop: each iteration takes $O(\diamG)$ time since flow is pushed along at most $2\diamG$ edges. Also, there are at most $n$ iterations, since each path saturates at least one vertex. Thus the while loop takes $O(n\diamG)$ time.
	\par For correctness, clearly $F \in \FE$; it remains to show the guarantee~\eqref{eq:roundtocirc}. Consider the path from $i$ to $v$ to $j$ along which we add flow in line~\ref{line:roundtocirc:pathij}. Since the paths from $i$ to $v$ and from $v$ to $j$ are both shortest paths, each is of length at most $\diamG$. Thus the total flow added to the path $i \to v \to j$ is at most $2 \diamG \delta_{ij}$.
Summing over all paths yields
\begin{align}
\|Q-P\|_1 
\leq 
\diamG \imbP.
\label{eq-pf:circ:1}
\end{align}
Now since $Q$ is entrywise bigger than $F$ and $P$, and since $\|F\|_1 = 1 = \|P\|_1$, we have $\|Q - F\|_1 = \|Q-P\|_1 \leq \diamG \imbP$. Therefore $\|F-P\|_1 \leq \|F-Q\|_1 + \|Q-P\|_1 
\leq 2 \diamG \imbP$.
\end{proof}

\subsubsection{Rounding to a quantized circulation}\label{sssec:round:qcirc}

We now address the quantization required in Lemma~\ref{lem:quantroundtocirc}: simply quantize and re-normalize $P$ before $\Roundtocirc$. Pseudocode is in Algorithm~\ref{alg:quantroundtocirc}. Note this quantization must be performed before $\Roundtocirc$ since quantizing afterwards can unbalance the circulation. Note also that we need an estimate of $\diamG$ for the quantization size; this is computed using the simple algorithm $\Approxdiam$ (see \S\ref{sec:prelim}).
The proof of Lemma~\ref{lem:quantroundtocirc} (i.e., the guarantee of $\Quantroundtocirc$) is straightforward from Lemma~\ref{lem:roundtocirc} (i.e., the guarantee of $\Roundtocirc$), and is deferred to Appendix~\ref{app:round:quantrountocirc}.

\begin{algorithm}
	\caption{
		$\Quantroundtocirc$:
		efficiently rounds to quantized circulation in $\FE$ without adjusting much flow. 
	}
	\hspace*{\algorithmicindent} \textbf{Input:} Weighted digraph $G = (V,E,w)$, normalized flow $P \in \DE$, accuracy $\eps$ \\
	\hspace*{\algorithmicindent} \textbf{Output}: Quantized, normalized circulation $F \in \FE$ satisfying~\eqref{eq:quantroundtocirc}
	\begin{algorithmic}[1]
		\State $\dt \gets \Approxdiam(G)$, $\alpha \gets \eps/(40 m \dt \Wmax)$\label{line:roundqcirc:diam}
		\State $R \gets \alpha \lfloor P/\alpha \rfloor$ \Comment{Round down $P_{ij}$ to integer multiple of $\alpha$}
		\label{line:roundqcirc:round}
		\State $\tilde{P} \gets R/\sum_{ij}R_{ij}$ \Comment{Renormalize to have unit total flow}
		\label{line:roundqcirc:normalize}
		\State \Return{$F \gets \Roundtocirc(G,\tilde{P})$} 
	\end{algorithmic}
	\label{alg:quantroundtocirc}
\end{algorithm}

\subsection{Rounding a circulation to a cycle}\label{ssec:round:extract}

Here we describe the algorithm $\Extractcycle$ and prove Lemma~\ref{lem:extractcycle}.
A simple approach for rounding a normalized circulation $F \in \FE$ to a cycle $\sigma$ satisfying $\wbar(\sigma) \leq \langle W, F\rangle$ is to decompose $F$ into a convex decomposition of cycles using the Cycle-Cancelling algorithm~\citep{Sch03}, and then output the cycle with best objective value. However, the runtime is a prohibitive $\Theta(mn)$. The bottleneck is $\Theta(m)$ cycle cancellations, each taking up to $\Theta(n)$ time. Intuitively, this factor of $n$ arises since cancelling a long cycle of length up to $n$ takes a long time yet does not give more ``benefit'' than a short cycle. We speed up this algorithm by exploiting the quantization of $F$ to ensure that cancelling long cycles gives a proportionally larger benefit than short cycles.
\par Specifically, let $\Extractcycle$ be the following minor modification of the Cycle-Cancelling algorithm. Initialize $\Ftilde = F$. While $\Ftilde \neq 0$, choose any vertex $i$ that has an outgoing edge $(i,j)$ with nonzero flow $\Ftilde_{ij} \neq 0$. Run Depth First Search (DFS) from $i$ until some cycle $\sigma$ is created. If $\wbar(\sigma) \leq \langle F, w \rangle$, then terminate. Otherwise, cancel the cycle $\sigma$ by subtracting $\tilde{f}_{\sigma} := \min_{e \in \sigma} \tilde{F}_{e}$ from the flow $\tilde{F}_{e}$ on each edge $e \in \sigma$. 
Then continue the DFS in a way that re-uses previous work---this is crucial for near-linear runtime. 
Specifically, if the previous DFS created a cycle by returning to an intermediate vertex $j \neq i$, then continue the DFS from $j$, keeping the work done by the DFS from $i$ to $j$. Otherwise, if the previous DFS created a cycle by returning to the initial vertex $i$, then restart the DFS at any vertex which has an outgoing edge with nonzero flow.
Note that $\Extractcycle$ leverages the quantization only in its runtime analysis.

\begin{proof}[Proof of Lemma~\ref{lem:extractcycle}]
Correctness is immediate by linearity. For the runtime, the key is the invariant that $\Ftilde$ remains a $\gamma$-quantized circulation. That $\Ftilde$ is a circulation ensures that the DFS always finds an outgoing edge and thus always finds a cycle since some vertex is eventually repeated. When such a cycle $\sigma$ is found, its cancellation lowers the total flow $\sum_{ij} \Ftilde_{ij}$ by $\tilde{f}_{\sigma} |\sigma|$, which is at least $\gamma|\sigma|$ by the invariant. Since the total flow is initially $\sum_{ij} F_{ij} = 1$, $\Extractcycle$ therefore terminates after cancelling cycles with at most $\gamma^{-1}$ total edges, counting multiplicity if an edge appears in multiple cancelled cycles. Since processing an edge takes $O(1)$ amortized time (again counting multiplicity), we conclude the desired $O(m + \gamma^{-1})$ runtime bound. 
\end{proof}

%% file: sections/alg_final.tex
\section{Concluding the approximation algorithm}\label{ssec:bal:final}

Algorithm~\ref{alg:bal} provides pseudocode for our proposed approximation algorithm $\Algbal$. It instantiates the framework in \S\ref{sec:fra} using the approximate Matrix Balancing reduction in Theorem~\ref{thm:bal:opt} for the optimization, and using the algorithm in Theorem~\ref{thm:round} for the rounding. By Theorem~\ref{thm:bal:opt}, $\Algbal$ succesfully approximates $\MMC$ regardless of how the balancing is performed. Since balancing is an active area of research (e.g.,~\citep{ZhuLiOliWig17,AltPar20bal,CohMadTsiVla17,OstRabYou16}), we abstract this computation into a subroutine $\ApproxBalance$: given a balanceable $K \in \Rpnn$ and an accuracy $\delta > 0$, $\ApproxBalance$ outputs a vector $x \in \Rn$ such that $\diag(e^x)$ solves approximate Matrix Balancing on $K$ to $\delta$ accuracy. 
Let $\Tbal(K,\delta)$ and $\Mbal(K,\delta)$ respectively denote the runtime and memory of $\ApproxBalance$.

\begin{algorithm}
	\caption{$\Algbal$: Matrix Balancing approach for approximating $\MMC$.}
	\hspace*{\algorithmicindent} \textbf{Input:} Weighted digraph $G = ([n],E,w)$, accuracy $\eps > 0$ \\
	\hspace*{\algorithmicindent} \textbf{Output:} Cycle $\sigma$ in $G$ satisfying $\wbar(\sigma) \leq \mAM(G) + \eps$
	\begin{algorithmic}[1]
		\Statex \textbackslash\textbackslash$\;$ Optimization step: compute near-feasible, near-optimal solution $P$ for~\eqref{MMC-P}
		\State $\dt \gets \Approxdiam(G)$, $\delta \gets \eps/(16\Wmax \dt)$ \Comment{Precision to balance}
		\State $\eta \gets 2.5 (\log m )/ \eps$, $K \gets \exp[-\eta W]$ \Comment{Matrix to balance}
		\State $x \gets \ApproxBalance(K,\delta)$, $A \gets \diag(e^x)K\diag(e^{-x})$, $P \gets A/(\sum_{ij} A_{ij})$ \Comment{Balance $K$}
		\label{line:Approxbalance}
		\Statex
		\Statex \textbackslash\textbackslash$\;$ Rounding step: round $P$ to a vertex of $\FE$ with nearly as good cost for~\eqref{MMC-P}
		\State $F \gets \Quantroundtocirc(G,P,\eps)$ \label{line:Roundtocirc} \Comment{Correct feasibility and quantize}
		\State $\sigma \gets \Extractcycle(G,F)$ \label{line:Extractcycle} \Comment{Round to vertex}
		\State \Return{$\sigma$}
	\end{algorithmic}
	\label{alg:bal}
\end{algorithm}

Below, \S\ref{ssec:bal:final:gen} establishes guarantees for $\Algbal$ in terms of a general subroutine $\ApproxBalance$, thereby reducing approximating $\MMC$ to approximate Matrix Balancing. In \S\ref{sssec:bal:final:concrete}, we implement $\ApproxBalance$ with concrete, state-of-the-art balancing algorithms to conclude our proposed $\MMC$ algorithm.

\subsection{Reducing $\MMC$ to matrix balancing}\label{ssec:bal:final:gen}

\subsubsection{Accuracy and runtime}\label{sssec:bal:final:time}

\begin{subtheorem}{theorem}\label{thm:mmc:reduction:all}
	\begin{theorem}[Accuracy and runtime of $\Algbal$]\label{thm:mmc:reduction:A}
		Given a weighted digraph $G = (V,E,w)$ and an accuracy $\eps > 0$, $\Algbal$ computes a cycle $\sigma$ in $G$ satisfying $\wbar(\sigma) \leq \mAMG + \eps$ in time $\Tbal(K,\delta) + O(m \diamG \Wmax/\eps)$.
	\end{theorem}
	\begin{proof}
		By the guarantee of $\Approxdiam$ (see \S\ref{sec:prelim}), $\diamG \leq \dt \leq 2 \diamG$.
		The runtime of $\Algbal$ follows from the runtimes of its constituent subroutines: $O(m)$ for $\Approxdiam$, and $O(m \diamG \Wmax/\eps)$ for rounding (Theorem~\ref{thm:round}). 
		Correctness follows from Observation~\ref{obs:fra} since $\Algbal$ implements both the optimization step (Theorem~\ref{thm:bal:opt}) and the rounding step (Theorem~\ref{thm:round}) to the accuracies prescribed in the algorithmic framework described in \S\ref{sec:fra} for $\delta = \eps/(16\Wmax \dt) \leq \eps/(16\Wmax \diamG)$.
	\end{proof}

	\subsubsection{Memory-optimality}\label{sssec:bal:final:mem}
	
	We now describe how to implement $\Algbal$ using only $O(n)$ additional memory. For ease of exposition, the memory usage counts the total numbers stored. (In \S\ref{sssec:bal:final:bits}, we show $\Algbal$ is implementable using $\tilde{O}(1)$-bit numbers.) Since storing $G$ requires $\Theta(m)$ memory, we assume $G=(V,E,w)$ is input to $\Algbal$ through two oracles:
	\begin{itemize}
		\item \underline{Edge oracle}: given $i \in V$ and $k \in [n]$, it returns the $k$-th incoming and outgoing edges from $i$ (in any arbitrary but fixed orders). If $k$ is larger than the indegree or outdegree of $i$, the respective query returns null.
		\item \underline{Weight oracle}: given $i,j \in V$, it returns $w(i,j)$ if $(i,j) \in E$, and $\infty$ otherwise.
	\end{itemize} 
	For simplicity, we assume that queries to these oracles take $O(1)$ time. In practice, the edge oracle can be implemented with simple, standard adjacency lists; and the weight oracle by e.g., hashing or re-computing weights on the fly if $w(\cdot,\cdot)$ is an efficiently computable function.
	
	\par Critically, in $\Algbal$ we do \emph{not} explicitly compute the intermediate matrices $K$, $A$, $P$, and $F$; instead, we form \emph{implicit} representations for them. To formalize this, it is helpful to define the notion of an $(T,M)$ \emph{matrix oracle} for a matrix: this is a data structure that uses $M$ storage, and can return a queried entry of the matrix in $T$ time and $O(1)$ additional memory.

	\begin{theorem}[Memory-optimality of $\Algbal$]\label{thm:mmc:reduction:B}
		There is an implementation of $\Algbal$ that, given $G$ through its edge and weight oracles, achieves the accuracy guarantee in Theorem~\ref{thm:mmc:reduction:A} and uses $\Tbal(K,\delta) + O(m\diamG \Wmax/\eps + m \log n)$ time and $\Mbal(K,\delta) + O(n)$ memory.
	\end{theorem}
	\begin{proof}
		We form an $(O(1), O(1))$ matrix oracle for $K$ by storing $\eta$---a query for entry $K_{ij}$ is performed by querying $w(i,j)$ and computing $e^{-\eta w(i,j)}$.  We form an $(O(1),O(n))$ matrix oracle for $P$ by storing $x$ and $s_A := \sum_{ij} e^{x_i - x_j} K_{ij}$---a query for entry $P_{ij}$ is performed by querying $K_{ij}$ and computing $e^{x_i - x_j} K_{ij} / s_A$. This matrix oracle for $P$ is passed as input to the rounding algorithms, which are implemented in the memory-efficient manner in Theorem~\ref{thm:round:mem}.
	\end{proof}

	\subsubsection{Bit-complexity}\label{sssec:bal:final:bits}

	Above, our analysis assumes exact arithmetic for ease of exposition; however, numerical precision is an important issue since na\"ively implementing $\Algbal$ can require large bit-complexity---indeed, since $\max_i x_i - \min_j x_j$ can be $\Omega(\diamG)$~\citep[\S3]{KalKhaSho97}, na\"ively operating on $A = \diag(e^x)K\diag(e^{-x})$ can require $\Omega(\diamG)$-bit numbers.
	Here, we establish that $\Algbal$ can be implemented on $\Otilde(1)$-bit numbers. (This analysis excludes the $\ApproxBalance$ subroutine since we have not yet instantiated it, but the concrete implementation used below also has logarithmic bit complexity; details in \S\ref{sssec:bal:final:concrete}.)
	
	\begin{theorem}[Bit-complexity of $\Algbal$]\label{thm:mmc:reduction:C}
		There is an implementation of $\Algbal$ that, aside from possibly $\ApproxBalance$, performs all arithmetic operations over $O(\log \tfrac{n \Wmax}{\eps}) = \tilde{O}(1)$-bit numbers and achieves the same runtime bounds (in terms of arithmetic operations), memory bounds (in terms of total numbers stored), and accuracy guarantees as in Theorem~\ref{thm:mmc:reduction:B}.
	\end{theorem}
\end{subtheorem}

This implementation essentially only modifies how $\Algbal$ computes entries of $K$, $A$, and $P$ on the exponential scale by using the log-sum-exp trick. Details are deferred to Appendix~\ref{app:mmc:reduction:bits}. Briefly, this modification relies on the observation that $\Algbal$ is robust in the sense that it outputs an $O(\eps)$-suboptimal cycle even if these entries are computed to low precision.

\subsection{Concrete implementation}\label{sssec:bal:final:concrete}

By Theorem~\ref{thm:mmc:reduction:all}, $\Algbal$ approximates $\MMC$ using any approximate balancing subroutine $\ApproxBalance$. The fastest practical instantiations of $\ApproxBalance$ are variants of Osborne's algorithm~\citep{Osborne60}. In particular, combining Theorem~\ref{thm:mmc:reduction:all} with the recent analysis of the Random Osborne algorithm in~\citep{AltPar20bal} yields the following near-linear runtime for approximating $\MMC$ on graphs with polylogarithmic diameter, both in expectation and with high probability. To emphasize the algorithm's practicality, below we write the single logarithmic factor in the runtime rather than hiding it with the $\Otilde$ notation.

\begin{theorem}[Main result: $\Algbal$ with Random Osborne]\label{thm:bal:orand}
	Consider implementing $\ApproxBalance$ using the Random Osborne algorithm in~\citep{AltPar20bal}. Then given a weighted digraph $G$ through its edge and weight oracles, and an accuracy $\eps > 0$, $\Algbal$ computes a cycle $\sigma$ in $G$ satisfying $\wbar(\sigma) \leq \mAMG + \eps$ using $O(n)$ memory and $T$ arithmetic operations on $O(\log(\tfrac{n \Wmax}{\eps})) = \Otilde(1)$-bit numbers, where $T$ satisfies
	\begin{itemize}
		\item (Expectation guarantee.) $\E[T] = O(m \diamG^2(\tfrac{\Wmax}{\eps})^2 \log n)$. 
		\item (High probability guarantee.) For all $\alpha \in (0,1)$, $\Prob\left( T \leq m \diamG^2(\tfrac{\Wmax}{\eps})^2 \log n \logalp \right) \geq 1 - \alpha$.
	\end{itemize}
\end{theorem}
\begin{proof}
	The runtime and bit-complexity of Random Osborne follow from~\citep[Theorem 5.1 and 8.1]{AltPar20bal} combined with the conditioning bound~\eqref{eq:bal:condK}. Random Osborne requires only $O(n)$ memory since $K$ is given through its query oracle. For the rest of $\Algbal$, apply Theorem~\ref{thm:mmc:reduction:all}.
\end{proof}

\begin{remark}[Numerical implementation]\label{rem:exp:bal-stable}
	As described in \S\ref{sssec:bal:final:bits}, computing $K = \exp[-\eta W]$ runs into numerical precision issues for large $\eta$. This is circumvented by \emph{not} explicitly computing $K$: numerical implementations of Osborne's algorithm operate on $K_{ij}$ only through $\log K_{ij} = -\eta W_{ij}$, and compute all intermediate quantities via the log-sum-exp trick~\citep{AltPar20bal}. 
\end{remark}

\begin{remark}[Alternative implementation]
	$\ApproxBalance$ can also be implemented using the algorithm of~\citep{CohMadTsiVla17}.
	This achieves comparable theoretical guarantees\footnote{Namely, $\tilde{O}(m \diamG (\Wmax/\eps)^3)$ arithmetic operations over $\Otilde(\poly(\Wmax/\eps))$-bit numbers (by combining Theorem 4.18 and Lemma 4.24 of~\citep{CohMadTsiVla17} with the bound~\eqref{eq:bal:condK}).}, but relies on Laplacian solvers which (currently) have no practical implementation.
\end{remark}

%% file: sections/experiments.tex
\section{Preliminary numerical simulations}\label{sec:experiments}
Although the focus of this paper is theoretical, here we provide preliminary numerics that investigate the practical aspects of our proposed algorithm $\Algbal$ and validate our theoretical findings.
\paragraph*{Experimental setup} We compared $\Algbal$ with state-of-the-art $\MMC$ algorithms on a number of different input graphs (e.g., sparse, dense, random, etc.). In all cases, we empirically observed that $\Algbal$ had close to linear runtime. Because many problem instances (e.g., random graphs) are ``easy'' for most MMC algorithms~\citep{GGTW09}, some competitor algorithms ran faster than expected on some of these inputs. Hence, in order to appreciate the differences between $\Algbal$ and the competitor algorithms, below we benchmark on the ``hardest'' families of problem instances from the comprehensive experimental survey~\citep{GGTW09}. These ``hard'' instances are formed by taking a random graph (either sparse or dense), planting a Hamiltonian Cycle and setting its weights so that it is the Minimum-Mean-Cycle, and then hiding this optimal cycle by randomly permuting the vertices and performing ``potential perturbations''; full reproduciblity details are provided in Appendix~\ref{app:exp}. The resulting graphs are either sparse (with $m \approx 7n$ edges) or dense (with $m \approx n^2/2$ edges), and have a unique Minimum-Mean-Cycle that is maximally long. 
All experiments are run on a standard 2018 MacBook Pro laptop.

\subsection{Scalability}
Figure~\ref{fig:scal} demonstrates that $\Algbal$ enjoys (close to) linear runtime in practice and is competitive with the three state-of-the-art algorithms implemented in the popular, heavily-optimized C++ LEMON library~\citep{lemon}. These competitors are the algorithm of Karp~\citep{Karp78}, the algorithm of Hartmann and Orlin~\citep{HarOrl93}, and the Howard iteration algorithm~\citep{CocCohGau98,Das04,DasIraGup99,Howard}. Note that $\Algbal$ computes approximate solutions whereas these competitors obtain exact solutions. In this experiment, the accuracy parameter of $\Algbal$ is set so that the suboptimality is $\mathord{\sim}10^{-3}$ (edge weights are normalized to $[0,1]$). Smaller $\eps$ leads to qualitatively similar results of near-linear runtime, although the constants of course degrade.

\begin{figure}
	\centering
	\begin{subfigure}{.5\textwidth}
		\centering
		\includegraphics[width=\linewidth]{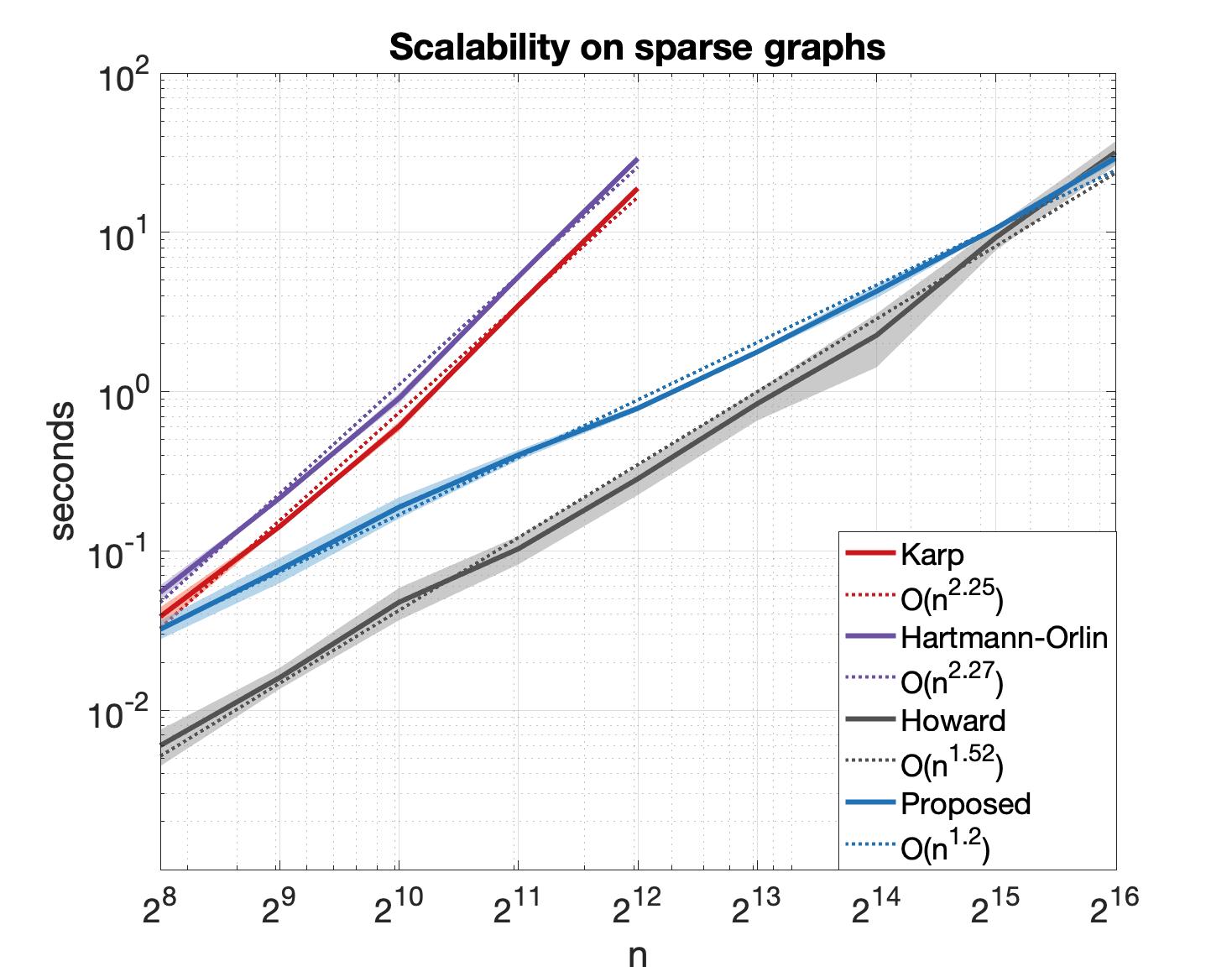}
		\caption{For sparse graphs with $m = \Theta(n)$ edges, a linear runtime is $O(m) = O(n)$.}
		\label{fig:scal:sparse}
	\end{subfigure}%
	\begin{subfigure}{.5\textwidth}
		\centering
		\includegraphics[width=\linewidth]{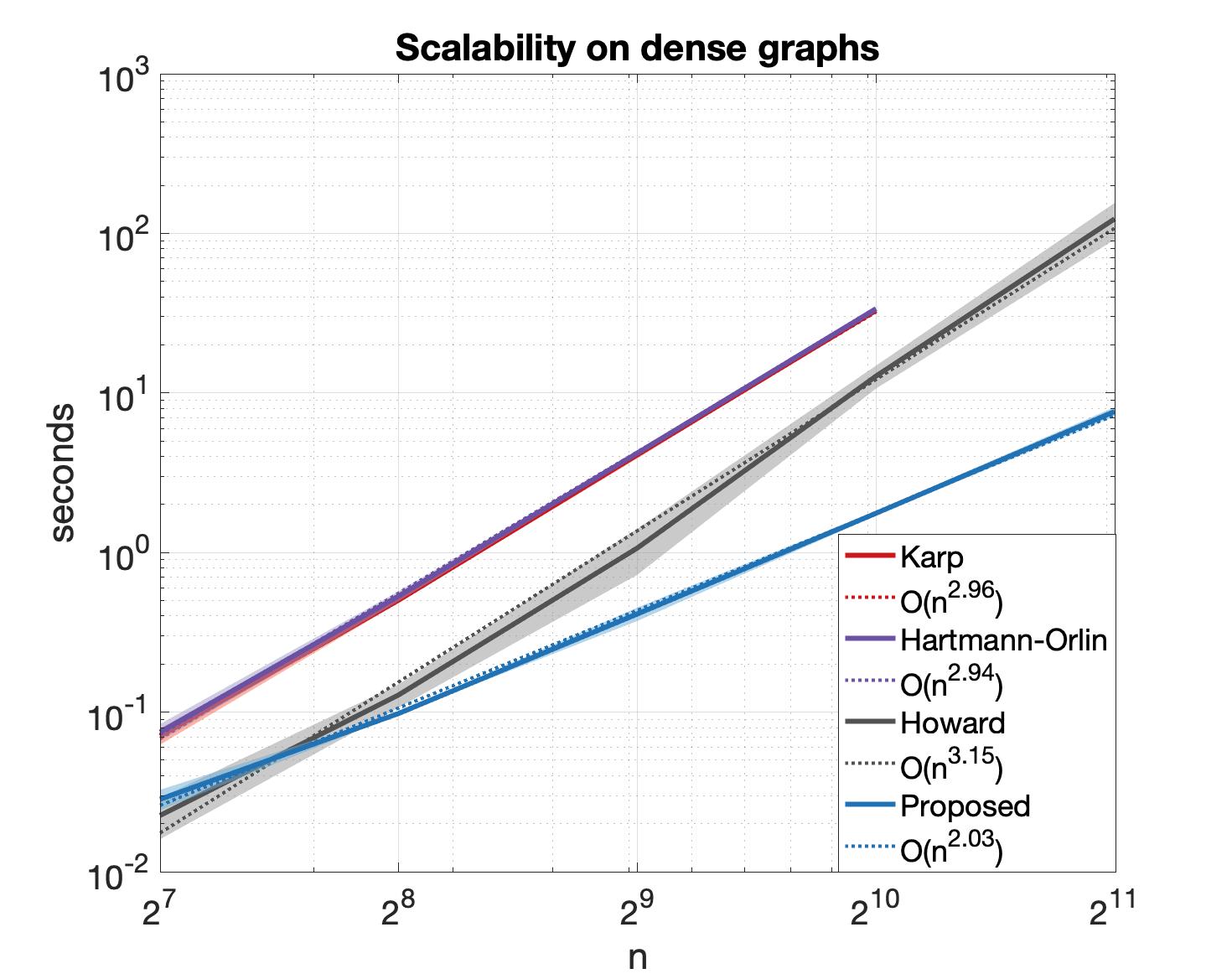}
		\caption{For dense graphs with $m = \Theta(n^2)$ edges, a linear runtime is $O(m) = O(n^2)$.}
		\label{fig:scal:dense}
	\end{subfigure}
	\caption{
		\small Scalability of our proposed algorithm $\Algbal$ vs state-of-the-art algorithms implemented in the popular LEMON library~\citep{lemon}. $\Algbal$ computes an approximate solution (here to roughly $3$ digits of precision) whereas the others compute exact solutions. The input instances are described in the main text.
	We report the average runtime (solid line) over $10$ runs, with $1$ standard deviation indicated by the shading. We estimate each algorithm's asymptotic runtime using linear regression (dashed line). The asymptotic runtime of $\Algbal$ on both sparse graphs (left) and dense graphs (right) is close to linear and outperforms all competitors.}
	\label{fig:scal}
\end{figure}

\par In Figure~\ref{fig:scal}, we estimate the asymptotic runtime of each algorithm using linear regression; these fits are quite accurate. Observe that $\Algbal$ has the fastest asymptotic runtime among all competitor algorithms. Moreover, the asymptotic runtime of $\Algbal$ on both the sparse graph inputs (Figure~\ref{fig:scal:sparse}) and dense graph inputs (Figure~\ref{fig:scal:dense}) is close to linear. In contrast, none of the competitor algorithms exhibit near-linear runtime scalings on either input. This enables $\Algbal$ to scale to larger instances than the competitor algorithms.
\paragraph*{Remarks about practical implementations of $\Algbal$}
Whereas the LEMON library is heavily-optimized, our implementation of $\Algbal$ is not. An optimized implementation of $\Algbal$ may lead to better constants and runtimes. Indeed, as written on page 1 of the empirical survey~\citep{GGTW09}, ``efficient implementations of MMC algorithms require nontrivial engineering, including data structures, efficient incremental restart, early termination detection, and hybrid algorithms.'' These are interesting directions for future research, but out of the scope of this paper.
\par It is worth pointing out that the sparse graphs used in the comparison in Figure~\ref{fig:scal:sparse} are particularly ``difficult'' inputs for our algorithm because these graphs have large (unweighted) diameter: this makes $\Algbal$ slower but does not similarly affect the known runtime bounds of the competitor algorithms. Nevertheless, $\Algbal$ outperforms the competitor algorithms in Figure~\ref{fig:scal:sparse} for large instances due to its faster asymptotic runtime. In practice, it is helpful to implement $\Algbal$ using the weighted diameter rather than $\Wmax$ times the unweighted diameter, since the former is smaller here; see the discussion in \S\ref{ssec:intro:contributions}.
\par We remark that we implement $\Algbal$ with a slightly different variant of Osborne's algorithm than in our theoretical results: Random-Reshuffle Cyclic Osborne (see~\citep{AltPar20bal} for a description). Random Osborne is used in our theoretical analysis and provably yields near-linear runtimes (Theorem~\ref{thm:bal:orand}). Random-Reshuffle Cyclic Osborne often enjoys slightly faster empirical convergence, but comparable theoretical guarantees are not known.

\subsection{Outperforming worst-case theoretical guarantees}

Here we mention that $\Algbal$ often finds significantly better approximations than our worst-case theoretical guarantees. 
A constant factor improvement is of course explained by the fact that we have not optimized the constants in this paper. However, even better performance appears to occur 
if the Cycle-Cancelling subroutine $\Extractcycle$ described in \S\ref{ssec:round:extract} is not terminated early; that is, if the fractional Matrix Balancing circulation is fully decomposed into cycles and the best one is output. The point is that often, at least one of these cycles is significantly better than the average---which is all that can be guaranteed in the worst-case by a linearity argument (c.f. \S\ref{ssec:round:extract}). Note also that our near-linear runtime bound still applies to this modified algorithm (since this is simply the worst-case of our proved runtime bound, c.f. the proof of Lemma~\ref{lem:extractcycle}).

\par To investigate the practical improvement from different versions of $\Extractcycle$, we plot in Figure~\ref{fig:gap} the error of three increasingly finer estimates of $\mAMG$ that $\Algbal$ (implicitly) makes:
\begin{itemize}
	\item ``Before rounding'' refers to the value $\langle W,F \rangle$ of the normalized circulation $F$ computed by $\Algbal$ before $\Extractcycle$ (i.e., the output of $\Quantroundtocirc$).
	\item ``Cancel fast'' refers to the value $\wbar(\sigma_{\textrm{fast}})$ of the cycle $\sigma_{\textrm{fast}}$ computed by the version of $\Extractcycle$ that terminates early.
	\item ``Cancel full'' refers to the value $\wbar(\sigma_{\textrm{full}})$ of the cycle $\sigma_{\textrm{full}}$ computed by the version of $\Extractcycle$ that does not terminate early.
\end{itemize}
Clearly, $\langle W, F\rangle \geq \wbar(\sigma_{\textrm{fast}}) \geq \wbar(\sigma_{\textrm{full}}) \geq \mAMG$. Indeed, each of these three estimates is an upper bound on $\mAMG$ by feasibility for the primal LP~\eqref{MMC-P}.
In Figure~\ref{fig:gap:subopt}, we plot this primal suboptimality, a.k.a., the difference between the estimate and $\mAMG$. Note that this suboptimality is not computable with $\Algbal$ since it requires the exact value of $\MMC$. In Figure~\ref{fig:gap:duality}, we plot an upper bound on this suboptimality that $\Algbal$ can provably certify: the duality gap between these primal estimates and the estimate of the dual LP~\eqref{MMC-D} obtained by using the approximate Matrix Balancing solution computed in step $1$ of $\Algbal$.

\begin{figure}
	\centering
	\begin{subfigure}{.5\textwidth}
		\centering
		\includegraphics[width=\linewidth]{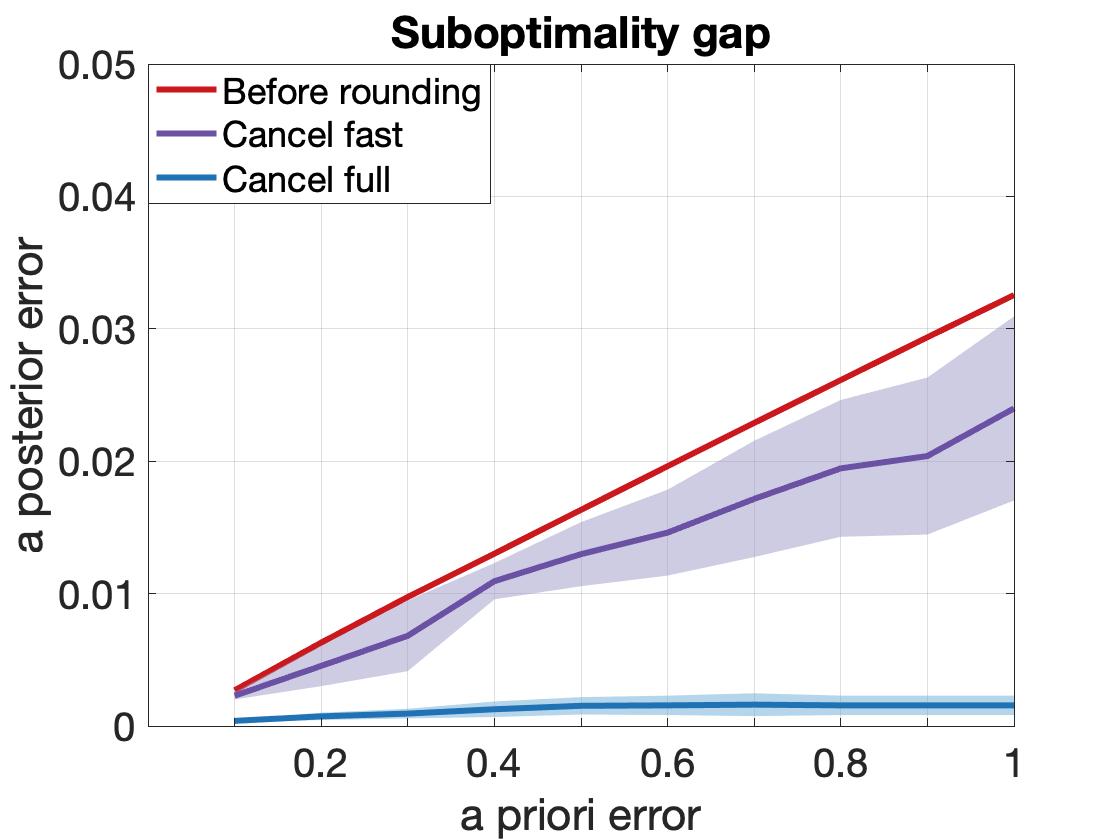}
		\caption{Error from the true value of $\MMC$.}
		\label{fig:gap:subopt}
	\end{subfigure}%
	\begin{subfigure}{.5\textwidth}
		\centering
		\includegraphics[width=\linewidth]{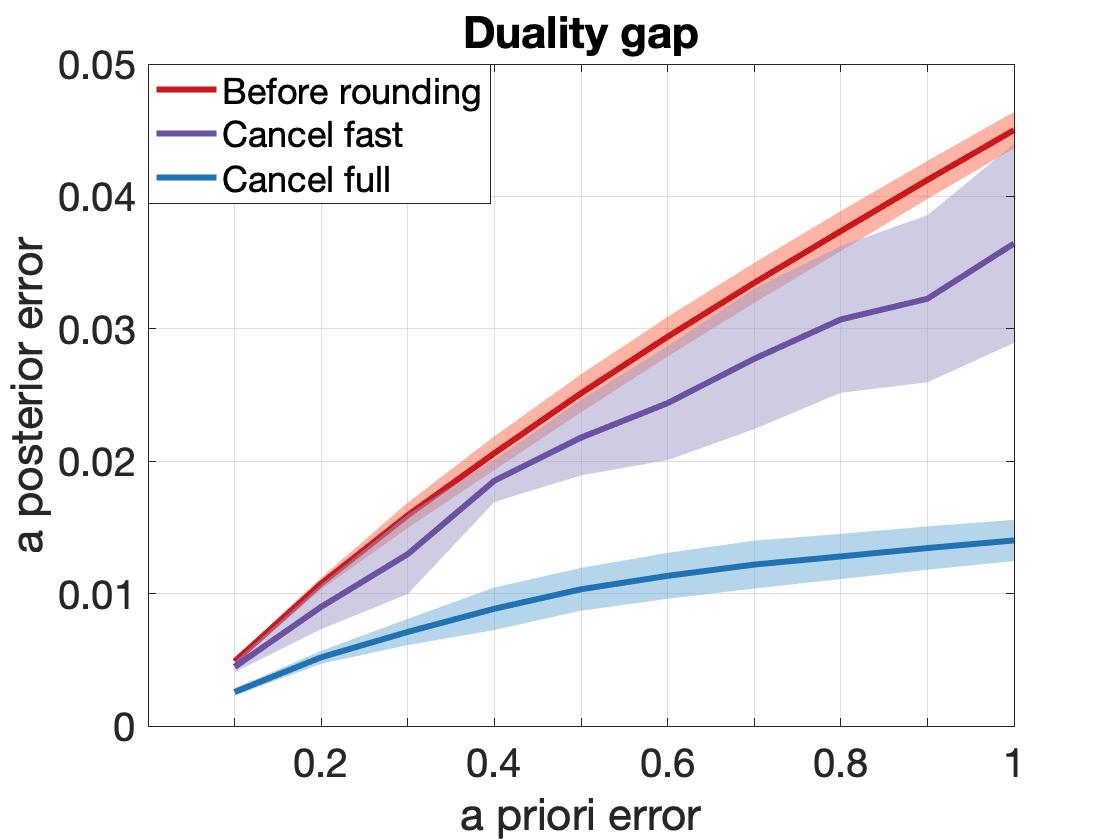}
		\caption{Error bound that $\Algbal$ can certify.}
		\label{fig:gap:duality}
	\end{subfigure}
	\caption{\small $\Algbal$ often finds significantly better approximations than our worst-case theoretical bounds guarantee. This is demonstrated by plotting the \emph{a posteriori} error versus the \emph{a priori} error estimate $\eps$. The \emph{a posteriori} error is measured via the suboptimality (left) and the duality gap (right). The input is the sparse graphs described in the main text, with $n = 2^{12}$ vertices. See the main text for a detailed description of the three plotted lines. We report the average performance over $50$ runs, with $1$ standard deviation indicated by the shading.}
	\label{fig:gap}
\end{figure}

\par As Figure~\ref{fig:gap} shows, in practice the error of $\Algbal$---measured either via the true suboptimality or the certifiable duality gap---is much better than the worst-case bounds when $\Extractcycle$ is terminated early, and moreover is even better when $\Extractcycle$ is run to completion.

%% file: sections/app.tex
\section{Deferred details}\label{app}

\subsection{Memory optimality of the rounding algorithm}\label{app:round:mem}

Here we describe a memory-efficient implementation of the rounding algorithm in Theorem~\ref{thm:round}. 
See \S\ref{sssec:bal:final:mem} for the definitions of a matrix oracle and the edge and weight oracles of a graph. 
Note that in what follows, $T = O(1)$ and $M = O(n)$ for $\Algbal$; see Theorem~\ref{thm:mmc:reduction:B}.

\begin{theorem}[Memory-efficient rounding]\label{thm:round:mem}
	If $G$ is given through its edge oracle and weight oracle, and $P$ is given through an $(T,M)$ matrix oracle, then the algorithm in Theorem~\ref{thm:round} can be run in $O( m ( T + \log n + \diamG \Wmax / \eps))$ time and $O(M + n )$ memory. 
\end{theorem}

\begin{proof}
	We describe how to implement the algorithms in Theorem~\ref{thm:round} in a memory-efficient way that does not change the outputted cycle. 
	The subroutine $\Approxdiam$ can be implemented using $O(n)$ memory since Breadth First Search can be implemented using the edge oracle for $G$ and $O(n)$ memory.
	To perform lines~\ref{line:roundqcirc:round} and~\ref{line:roundqcirc:normalize}, $\Quantroundtocirc$ forms an $(T+O(1),M+O(1))$ matrix oracle for $\tilde{P}$ by using $O(1)$ additional memory to compute and store $s_R := \sum_{ij} R_{ij}$---then an entry $\tilde{P}_{ij}$ can be queried by querying $P_{ij}$ and computing $\alpha \lfloor P_{ij} / \alpha \rfloor / s_R$.
	\par $\Roundtocirc$ takes this matrix oracle for $\tilde{P}$ as input and forms an $(T+O(\log n), M+O(n))$ matrix oracle for $F$. Specifically, it implicitly performs line~\ref{line:roundtocirc:pathij} by storing in a Balanced Binary Search Tree the amount of flow, totalled over these saturating paths, pushed along each edge. This takes $O(n)$ additional storage since all edges lie on the Shortest Paths trees in or out of $v$, which collectively contain at most $2(n-1)$ edges. The matrix oracle for $F$ also stores $s_Q := \sum_{ij}Q_{ij}$---then an entry $F_{ij}$ can be queried by querying $\tilde{P}_{ij}$, querying the amount of adjusted flow on edge $(i,j)$ in the Balanced Binary Search Tree, and re-normalizing by $s_Q$.
	\par In $\Extractcycle$, we maintain for each vertex $i$ a counter $j_i$. This is the lowest index with respect to the (outgoing) edge oracle of $G$, that corresponds to an outgoing edge from $i$ with nonzero flow.
	The DFS always takes these edges. We query each $F_{ij}$ at most once: the first time we cancel a cycle with that edge. If the edge is partially cancelled, then we store the remaining flow. (If the edge is fully saturated, then we do not need to store anything since we will never come back to it). By the bias of the DFS, there are always at most $n$ partially cancelled edges (one for each vertex), so this requires $O(n)$ additional memory. 
\end{proof}

\subsection{Proof of Lemma~\ref{lem:quantroundtocirc}}\label{app:round:quantrountocirc}

\begin{lemma}[Helper lemma for $\Quantroundtocirc$]\label{lem:round:trunc}
	Consider $P$, $R$, $\tilde{P}$, and $\alpha$ in $\Quantroundtocirc$. Then (i) $\|\tilde{P}-P\|_1 \leq 2 \alpha m$, and (ii) $\imb{\tilde{P}} \leq 2\imbP + 4\alpha m$.
\end{lemma}
\begin{proof}
	\par Proof of item (i). First note that since rounding $P$ to $R$ changes every entry by at most $\alpha$, thus
	$
		\|R - P\|_1 \leq \alpha m$, and so also 
		$\sum_{ij} R_{ij}
		\geq 1 - \alpha m$.
	By definition of $\tilde{P}$, $\|\tilde{P} -R\|_1 = 1 - \|R\|_1 \leq \alpha m$. Thus by the triangle inequality, $\|\tilde{P} - P\|_1 \leq \|\tilde{P} - R\|_1 + \|R - P\|_1 \leq 2\alpha m$. 
	\par Proof of item (ii). Note that rounding on an edge to an integer multiple of $\alpha$ increases the flow imbalance at each adjacent vertex by at most $\alpha$, thereby increasing the total imbalance by at most $2\alpha$. Thus $R$ has imbalance at most $\imb{R} \leq \imbP+ 2\alpha m$. By definition of $\tilde{P}$, we have $\imb{\tilde{P}} = \imb{R} / (\sum_{ij}R_{ij}) \leq (\imbP + 2\alpha m) / (\sum_{ij}R_{ij}) $. We therefore conclude by observing that $1/(\sum_{ij} R_{ij}) \leq 2$, which follows from $\sum_{ij} R_{ij}
	\geq 1 - \alpha m$ combined with the fact that $\alpha \leq 1/(2m)$.
\end{proof}

\begin{proof}[Proof of Lemma~\ref{lem:quantroundtocirc}]
	The runtime bound follows from the runtimes of $\Approxdiam$ (see \S\ref{sec:prelim}) and $\Roundtocirc$ (Lemma~\ref{lem:roundtocirc}). The guarantee $F \in \FE$ is immediate from Lemma~\ref{lem:roundtocirc}.
	\par Next, we establish~\eqref{eq:quantroundtocirc}. By item (i) of Lemma~\ref{lem:round:trunc}, $\|\tilde{P} - P\|_1 \leq 2\alpha m$. Moreover, by Lemma~\ref{lem:roundtocirc} and then item (ii) of Lemma~\ref{lem:round:trunc}, $\|F - \tilde{P}\|_1 \leq 2 \diamG \imbPtilde \leq 4 \diamG \imbP+ 8 \alpha m \diamG$. Thus $\|F-P\|_1 \leq \|F - \tilde{P}\|_1 + \|\tilde{P} - P\|_1 \leq 4 \diamG \imbP + 10\alpha m \diamG$. By our choice of $\alpha$ and the bound $\dt \geq \diamG$ (see \S\ref{sec:prelim})), the latter summand is at most $\eps/(4\Wmax)$.  
	\par Finally, we establish the quantization guarantee. By construction, $R$ is $\alpha$-quantized, and so $\tilde{P}$ is $\beta$-quantized for $\beta := \alpha/(\sum_{ij} R_{ij}) \geq \alpha$. 
	Since $\tilde{P}$ is the input to $\Roundtocirc$ in $\Quantroundtocirc$, in $\Roundtocirc$ $Q$ will be $\beta$-quantized since $\tilde{P}$ is. Thus $F$ is $\gamma$-quantized for $\gamma := \beta/\sum_{ij}Q_{ij}$. Now $\sum_{ij} Q_{ij} = \sum_{ij} \tilde{P}_{ij} + \sum_{ij}( Q_{ij} - \tilde{P}_{ij} ) \leq 1 +  \diamG \imbPtilde$ by~\eqref{eq-pf:circ:1}, and this is $O(1)$ by item (ii) of Lemma~\ref{lem:round:trunc} and the assumption that $\imbP \leq 1/\diamG$.
	Therefore $\gamma = \Omega(\beta) = \Omega(\alpha)$. We conclude by our choice of $\alpha$ and the fact that $\dt \leq 2\diamG$ (see \S\ref{sec:prelim}).
\end{proof}

\subsection{Bit complexity}\label{app:mmc:reduction:bits}

Here we prove Theorem~\ref{thm:mmc:reduction:C}. For simplicity of exposition, we omit constants and show how to ensure $\Algbal$ outputs an $O(\eps)$-suboptimal cycle; the claim then follows by re-normalizing $\eps$.

\begin{proof}[Proof of Theorem~\ref{thm:mmc:reduction:C}]
	\underline{Modification of $\Algbal$.} The computation of $A$ and $P$ is modified slightly as follows. Let $\alpha = c\eps/(\Wmax m \diamG)$ for a sufficiently small constant $c$. (i) Read and store the input weights $W_{ij}$ and the output $x$ of $\ApproxBalance$ to $\plusminus \alpha$ precision. (ii) Compute and store $Y_{ij} := x_i - x_j + \eta W_{ij}$ to $\plusminus \alpha$ precision for each $(i,j) \in E$. (iii) Translate $Z_{ij} := Y_{ij} - y$, where $y := \max_{ij} Y_{ij}$. (iv) Compute $A_{ij} = e^{Z_{ij}}$ to $\plusminus \alpha^2$ precision if $Z_{ij} \geq \log \alpha$, and set $A_{ij} = 0$ otherwise. (v) Compute entries of $P = A / \sum_{ij} A_{ij}$ to $\plusminus \alpha$ precision.

	\par \underline{Bit-complexity analysis.} By definition of $\alpha$, $\log \frac{1}{\alpha} = 
	O(\log \tfrac{n\Wmax}{\eps})
	= \Otilde(1)$. (i) The bit complexity of the stored weights is thus $O(\log \tfrac{\Wmax}{\alpha}) = \Otilde(1)$. The bit complexity of the stored $x$ is $O(\log \tfrac{\max_i x_i - \min_i x_i}{\alpha}) = \Otilde(1)$, since $\log (\max_i x_i - \min_i x_i ) = 
	O(\log \tfrac{n\Wmax}{\eps})
	= \Otilde(1)$ by Lemma~\ref{lem:bal:R} and~\eqref{eq:bal:condK}. (ii), (iii) The bit complexity of $Y$, $y$, $Z$ is similarly $\Otilde(1)$. (iv) The bit complexity of $A_{ij}$ is $O(\log \tfrac{1}{\alpha}) = \Otilde(1)$. (v) The bit complexity of $P_{ij}$ is $O(\log \tfrac{1}{\alpha^2}) = \Otilde(1)$. Since $P$ has low bit-complexity, the rest of $\Algbal$ does by construction of the rounding algorithms.
	
	\par \underline{Proof of correctness.} We make use of the following lemma.
	\begin{lemma}[Robustness of $\Algbal$]\label{lem:bal:bit-rob}
		The following changes to $\Algbal$ affect the mean-weight $\wbar(\sigma)$ of the returned cycle $\sigma$ by at most $\plusminus O(\eps)$:
		\begin{itemize}
			\item [(1)] The entries of $P$ are approximated to $\plusminus \alpha$ additive error and remain nonnegative.
			\item [(2)] The nonzero entries of $P$ are approximated to $[1\plusminus \alpha]$ multiplicative error.
			\item [(3)] The nonzero entries of $A$ are approximated to $[1\plusminus \alpha]$ multiplicative error.
		\end{itemize}
	\end{lemma}
	\begin{proof}
		The proof of item (1) is identical to the truncation in $\Quantroundtocirc$ in Lemma~\ref{lem:quantroundtocirc}.
		Item (2) then follows since $P_{ij} \leq 1$. Item (3) then follows since $P = A / (\sum_{ij}A_{ij})$. 
	\end{proof}
	By the guarantee for $\Algbal$ in exact arithmetic (Theorem~\ref{thm:mmc:reduction:A}), it suffices to show that these modifications (i)-(v) affect $\wbar(\sigma)$ by at most $\plusminus O(\eps)$. (i) and (ii) change $A_{ij}$ by $[1\plusminus O(\alpha)]$ multiplicative error, which is acceptable by item (3) of Lemma~\ref{lem:bal:bit-rob}. (iii) rescales $A$, which does not alter $P$. (iv) First we argue the effect of dropping all $A_{ij} < \alpha$ to $0$. The only affected entries of $P_{ij}$ are those dropped to $0$; and since (iii) ensures $\sum_{i'j'} A_{i'j'} \geq \max_{i'j'} A_{i'j'} = 1$, thus $P_{ij} = A_{ij} / \sum_{i'j'} A_{i'j'}$ must have been at most $\alpha$, so setting $P_{ij}$ to $0$ is acceptable by item (1) of Lemma~\ref{lem:bal:bit-rob}. Next, we argue the truncation of $A_{ij}$. The $\plusminus \alpha^2$ additive precision of $A_{ij}$ implies $[1\plusminus \alpha]$ multiplicative error for the nonzero entries of $A$ (since they are at least $\alpha$), which is acceptable by item (3) of Lemma~\ref{lem:bal:bit-rob}. Finally, (v) is acceptable by item (1) of Lemma~\ref{lem:bal:bit-rob}.
\end{proof}

\subsection{Reproducibility details for the experiments}\label{app:exp}

Both the sparse and dense inputs used in \S\ref{sec:bal} are generated in a three-step process \'a la the experimental survey~\citep{GGTW09}. First, the underlying graph is generated. For the dense graphs, this is an Erd\"os-Renyi random graph where each edge is included with probability $1/2$ and has uniform random weights in $\{1,\dots,100\}$. For the sparse graphs, this is a random graph with $5n$ random edges and a random Hamiltonian cycle, again all with uniform random weights in $\{1,\dots,100\}$. Second, we plant a Hamiltonian cycle that has weight $-1$ on one edge, and weight $0$ on the rest. This is the ``subfamily 05'' perturbation of~\citep{GGTW09}. It ensures that graph has a unique Minimum-Mean-Cycle and moreover that this optimal cycle is maximally long. Third, the planted Hamiltonian Cycle is hidden by randomly permuting the vertices and performing a ``potential perturbation''; that is, adjusting $w(i,j) \mapsto w(i,j) + p_i - p_j$ where $p \in \R^n$ is a random vector with entries drawn uniformly from $\{1,\dots,200\}$. This potential perturbation does not affect the Minimum Mean Cycle. Finally, all edge weights are normalized to $[0,1]$ via a simple shift and scaling.